%% file: aaai25-automata.tex
\documentclass[letterpaper]{article}
\usepackage{aaai25} 
\usepackage{times} 
\usepackage{helvet} 
\usepackage{courier} 
\usepackage[hyphens]{url} 
\usepackage{graphicx}
\usepackage{natbib} 
\usepackage{caption}
\usepackage{algorithm}
\usepackage{algpseudocode}
\usepackage[inline]{enumitem}
\usepackage{amssymb}
\usepackage{amsthm}
\usepackage{thmtools}
\usepackage{mathtools}
\usepackage[capitalise,noabbrev]{cleveref}
\usepackage[all]{foreign}
\usepackage{xcolor}
\usepackage{marginnote}

\usepackage{packages/commands}
\usepackage{packages/symbols}
\usepackage{packages/logic}

\nocopyright

\newtheorem{proposition}{Proposition}
\newtheorem{definition}{Definition}
\newtheorem{lemma}{Lemma}
\newtheorem{theorem}{Theorem}
\newtheorem{encoding}{Encoding}
\newtheorem{example}{Example}

\crefname{section}{Section}{Sections}
\Crefname{section}{Section}{Sections}
\crefname{subsection}{Subsection}{Subsections}
\Crefname{subsection}{Subsection}{Subsections}
\crefname{encoding}{Encoding}{Encodings}
\Crefname{encoding}{Encoding}{Encodings}

\newtheoremstyle{claimstyle}%
  {3pt}{3pt}{\normalfont}{}%
  {}{.}{0.5em}%
  {\underline{#1\ #2}}
\theoremstyle{claimstyle}
\newtheorem{claim}{Claim}[section]
\newtheorem*{claim*}{Claim}
\crefname{claim}{Claim}{Claims}
\Crefname{claim}{Claim}{Claims}
\newenvironment{claimproof}[1]{\par\noindent\underline{Proof:}\space#1}{\hfill $\blacksquare$}

\newif\ifappendix
\appendixtrue

\IfFileExists{config.cfg}{\input{config.cfg}}{}

\usepackage{newfloat}
\usepackage{listings}
\DeclareCaptionStyle{ruled}{labelfont=normalfont,labelsep=colon,strut=off} 
\floatstyle{ruled}
\newfloat{listing}{tb}{lst}{}
\floatname{listing}{Listing}

\title{%
  A general automata model for first-order temporal logics (extended version)
}
\author{
  Luca Geatti\textsuperscript{\rm 1},
  Alessandro Gianola\textsuperscript{\rm 2},
  Nicola Gigante\textsuperscript{\rm 3}
}

\affiliations{
  \textsuperscript{\rm 1}Dipartimento di Scienze Matematiche, Informatiche e Fisiche, University of Udine, Udine, Italy\\
  \textsuperscript{\rm 2}INESC-ID/Instituto Superior T{\'e}cnico, Universidade de Lisboa, Lisbon, Portugal\\
  \textsuperscript{\rm 3}Faculty of Engineering, Free University of Bozen-Bolzano, Bolzano, Italy \\
  luca.geatti@uniud.it, alessandro.gianola@tecnico.ulisboa.pt, nicola.gigante@unibz.it
}

\begin{document}

\maketitle

\begin{abstract}
    First-order linear temporal logic (\FOLTL) is a flexible and expressive
    formalism capable of naturally describing complex behaviors and properties.
    Although the logic is in general highly undecidable, the idea of using it as
    a specification language for the verification of complex infinite-state
    systems is appealing. However, a missing piece, which has proved to be an
    invaluable tool in dealing with other temporal logics, is an
    \emph{automaton} model capable of capturing the logic. In this paper we
    address this issue, by defining and studying such a model, which we call
    \emph{first-order automaton}. We define this very general class of automata,
    and the corresponding notion of \emph{regular} first-order language (of
    \emph{finite} words), showing their closure under most language-theoretic
    operations. We show how they can capture any \FOLTL formula over finite
    words, over any signature and theory, and provide sufficient conditions for
    the \emph{semi-decidability} of their non-emptiness problem. Then, to show
    the usefulness of the formalism, we prove the decidability of \emph{monodic}
    \FOLTL, a classic result known in the literature, with a simpler and direct
    proof.
\end{abstract}

\input{sections/1.introduction.tex}
\input{sections/2.preliminaries.tex}
\input{sections/3.languages.tex}
\input{sections/4.automata.tex}
\input{sections/5.encoding.tex}

\input{sections/6.purepast.tex}
\input{sections/7.emptiness.tex}
\input{sections/8.decidable.tex}
\input{sections/9.conclusions.tex}

\section*{Acknowledgements}
Luca Geatti acknowledges the support from the Interconnected Nord-Est
Innovation Ecosystem (iNEST), which received funding from the European
Union Next-GenerationEU (PIANO NAZIONALE DI RIPRESA E RESILIENZA (PNRR)
– MISSIONE 4 COMPONENTE 2, INVESTIMENTO 1.5 – D.D. 1058 23/06/2022,
ECS00000043).
Alessandro Gianola was also partly supported by Portuguese national funds through Fundação para a Ciência e a Tecnologia, I.P. (FCT) under projects UID/50021/2025 (DOI: https://doi.org/10.54499/UID/50021/2025) and UID/PRR/50021/2025 (DOI: https://doi.org/10.54499/UID/PRR/50021/2025).

\bibliography{biblio.bib}

\ifappendix
    \appendix

    \input{sections/a.proofs.tex}
    \input{sections/b.encodings.tex}
\fi

\end{document}

%% file: sections/1.introduction.tex

\section{Introduction}
\label{sec:introduction}

Linear temporal logic (\LTL)~\cite{Pnueli77} is the most common specification
language for the verification of reactive systems. \LTL on \emph{finite words}
(\LTLf, \cite{DeGiacomoV13}) also gained traction recently in many applications
in AI and formal verification. However, \LTL and \LTLf are \emph{propositional}
logics, and as such they are limited to the specification of \emph{finite-state}
systems, whereas many real-world scenarios require modeling an infinite state
space. 

First-order \LTL (\FOLTL) is a temporal logic that mixes \LTL with classical
first-order logic, resulting in an extremely expressive and flexible formalism.
The logic has been studied quite extensively in the
literature~\cite{GabbayKWZ03,KontchakovLWZ04}, especially in the field of
\emph{knowledge representation}, where it is used as the formal underpinning of
temporal description logics, temporal databases, \etc. However, reasoning about
\FOLTL is quite hard: in general, satisfiability of \FOLTL sentences is highly
undecidable, \ie not even recursively enumerable~\cite{GabbayKWZ03}.
Nevertheless, the idea of using well-behaved fragments of \FOLTL as a
specification language for verification of infinite-state systems is appealing
and worth investigating. We are interested in tasks such as satisfiability and
model checking~\cite{ClarkeGKPV18} of \FOLTL specifications, but also
\emph{reactive synthesis}~\cite{PnueliR89}, \ie building a \emph{strategy} that
fulfils a \FOLTL specification in adversarial environments, and
\emph{monitoring}, \ie detecting violations of a specification during a system's
execution~\cite{LeuckerS09}. Indeed, the MoXI intermediate language for model
checking features a semantics based on
\FOLTL~\cite{JohannsenNDISTVR24,RDIJNSTV24}, specifically to support
infinite-state model checking.

However, to make \FOLTL a really practicable specification language, we miss
some ingredients. One of these is the definition of a model of \emph{automata}
capable of capturing \FOLTL, and a theory of the \emph{languages} (sets of
models) that \FOLTL can describe. These ingredients have been crucial in the
development of the theory and practice of finite-state verification, with formal
language theory being the theoretical background~\cite{DBLP:books/wi/LoeckxS87}
and automata providing the algorithmic
powerhouse~\cite{vardi2005automata,vardi1991verification}.

In this paper, we define and study a new kind of infinite-state, finitely
representable automata on \emph{finite words}, that we call \emph{first-order
automata}, capable of capturing \FOLTL on finite words. In particular, our
contributions are the following:
\begin{enumerate}
  \item we introduce the definition of first-order automata and the associated
    notion of \emph{first-order regular languages};
  \item we study the \emph{closure properties} of these languages, \ie that they
    are closed by union, intersection, concatenation, and Kleene star, and we
    discuss deterministic first-order automata (although determinization is
    still open);
  \item we show how first-order automata can easily and directly encode any 
    \FOLTL sentence over any first-order theory;
  \item we show, similarly to what happens for propositional \LTL, that if we
    start with an \FOLTL sentence that uses only \emph{past operators}, the
    corresponding automaton is \emph{deterministic}; as determinism is essential
    to solve \emph{reactive synthesis} and \emph{monitoring} for \LTL
    specifications~\cite{ArtaleGGMM23}, this result is encouraging;
  \item \label{contrib:semidecidability} 
    we address the non-emptiness problem for first-order automata (\ie the
    existence of any accepted word), providing some sufficient conditions for
    its \emph{semi-decidability};
  \item\label{contrib:decidability} we discuss some sufficient conditions for
    \emph{decidability} of the emptiness problem, including a class of
    \emph{monadic} automata which are the natural counterpart of \emph{monodic}
    \FOLTL sentences (\ie the fragment where any temporal subformula has only
    one free variable).
\end{enumerate}
As evidence of the usefulness of automata in dealing with \FOLTL,
\cref{contrib:decidability} provides a self-contained and more elementary proof of the
decidability of \emph{monodic} \FOLTL on finite words. This is a classic
result~\cite{GabbayKWZ03} whose original proof, however, is conceptually quite complex, since it relies on the
decidability result of monadic second order logic over $(\N,<)$ by
Büchi~\cite{Buchi62}.

We proceed as follows. \Cref{sec:languages} introduces first-order temporal
logic and the first-order languages it is interpreted over, and
\cref{sec:automata} defines first-order automata and first-order \emph{regular}
languages, studying their closure properties. \Cref{sec:encoding} shows how
first-order automata capture \FOLTL, and \cref{sec:logic:purepast} shows how
pure-past sentences result in deterministic automata, discussing the
implications of this construction in applications. \Cref{sec:emptiness}
discusses the non-emptiness problem, and \cref{sec:decidable} discusses some
decidable cases, including our self-contained proof of the decidability of
monodic \FOLTL.  \Cref{sec:conclusions} concludes with a discussion and related
work. \ifappendix 
All the proofs are in the Appendix.
\else
All the proofs can be found in the extended version of the
paper \cite{ExtendedVersionAAAI}.
\fi 

%% file: sections/2.preliminaries.tex

\section{Preliminaries}
\label{sec:preliminaries}

We consider classical \emph{many-sorted} first-order logic.

A \emph{many-sorted} first-order \emph{signature} is a set $\Sigma$ of
\emph{symbols} that can be \emph{sort} symbols, \emph{function} symbols, or
\emph{relation} symbols. The subset of sort symbols in $\Sigma$ is denoted
as $\sorts(\Sigma)$. Function symbols $f$ come with an arity
$\arity(f)\in\N$, an input sort $\sort_i(f)\in\sorts(\Sigma)$ for each
$0\le i < \arity(f)$, and an output (range) sort
$\range(f)\in\sorts(\Sigma)$. Similarly we have $\arity(R)$ and
$\sort_i(R)$ for each relation symbol $R$. A \emph{constant} is a function
symbol $c$ with $\arity(c)=0$. We consider only \emph{finite} signatures.

Given a many-sorted signature $\Sigma$, a \emph{$\Sigma$-structure} is a tuple
of maps $\sigma=(\sigma_s, \sigma_f, \sigma_r)$ where $\sigma_s$ maps each sort
$S\in\Sigma$ to a domain set $S^\sigma$, $\sigma_f$ maps each function symbol
$f\in\Sigma$ of arity $n$ to a function $f^\sigma$ from
$\sort_0(f)^\sigma\times\ldots\times\sort_{n-1}(f)^\sigma$ to
$\range(f)^\sigma$, and $\sigma_r$ maps each relation symbol $R\in\Sigma$ of
arity $n$ to a relation
$R^\sigma\subseteq\sort_0(R)^\sigma\times\ldots\times\sort_{n-1}(R)$. We allow
relation symbols of arity zero, which we call \emph{propositions}. Given a
structure $\sigma$, a proposition $p$ can have only two possible
interpretations: $p^\sigma=\emptyset$, meaning that $p$ is false in $\sigma$, or
$p^\sigma=\set{()}$, meaning that $p$ is true in $\sigma$. The
\emph{size} of a $\Sigma$-structure $\sigma$, denoted $\abs{\sigma}$, is the
maximum cardinality $\abs{S^\sigma}$ among all the sorts $S\in\sorts(\Sigma)$.

Given a signature $\Sigma$, a $\Sigma$-theory $\theory$ is a set of sentences
over $\Sigma$ ($\Sigma$-sentences). A sentence $\phi$ is
\emph{$\theory$-satisfiable} if and only if there exists a $\Sigma$-structure
satisfying $\phi$ and all sentences in $\theory$.

An isomorphism between two $\Sigma$-structures is a structure-preserving
bijection from the domain of one $\Sigma$-structure to the domain of the other,
whose inverse is also structure-preserving. This implies that the two
$\Sigma$-structures can be treated as identical up to a renaming of elements.

The downwards \emph{Löwenheim-Skolem} theorem, extended to many-sorted
first-order logic, states that if a theory $\theory$ admits a model of size
$\kappa\ge\aleph_0$ then it has an elementarily equivalent model of size
$\aleph_0$. Then, up to isomorphism, we can consider \emph{w.l.o.g}\ only models
of $\theory$ where each finite sort $S$ of cardinality $n$ is interpreted as
$\set{0,\ldots,n-1}$ and each infinite sort is interpreted as $\N$ itself. This
allows us to choose a canonical (up to isomorphism) representative for each
class of elementarily equivalent $\theory$-models, making this class a standard
set\footnotemark. Hence, we properly get the \emph{set} of all such models of a
theory, which we denote as $\sem{\theory}$. We remark that $\sem{\theory}$
contains only $\Sigma$-structures of cardinality at most countable. In the
following, by $\theory$-structure we mean an element of $\sem{\theory}$.

\footnotetext{%
  Without restricting the underlying domains, the collection of all models of a
  theory would generally form a proper class rather than a set. Indeed, even for
  the empty theory over a signature with a single sort, considering all possible
  singleton domains would amount to considering every set as an element of some
  model's domain. If the collection of all such models were itself a set, the
  union of their domains would yield the set of all sets, leading to Russell's
  paradox.%
}

Given a $\Sigma$-structure $\sigma$ and a signature $\Gamma$ such that
$\Sigma\subseteq\Gamma$, a $\Gamma$-expansion of $\sigma$ is a
$\Gamma$-structure that agrees with $\sigma$ on the symbols from $\Sigma$.

%% file: sections/3.languages.tex

\section{First-Order Languages and \\Temporal Logic}
\label{sec:languages}

\textbf{First-Order Words and Languages.}
Given a signature $\Sigma$ we assume that each symbol $s\in\Sigma$ comes with a
flag $\rigid(s)\in\set{\top,\bot}$ telling whether the symbol is \emph{rigid} or
\emph{non-rigid}, respectively. The set of the former is denoted as
$\rigid(\Sigma)\subseteq\Sigma$. We assume sort symbols are \emph{always rigid},
so $\sorts(\Sigma)\subseteq\rigid(\Sigma)$. Given a signature $\Sigma$, we
denote as $\Sigma'$ the signature obtained from $\Sigma$ by replacing any
\emph{non-rigid} symbol $s$ with a primed version $s'$. Similarly, given a
$\Sigma$-structure $\sigma$ we denote as $\sigma'$ the corresponding
$\Sigma'$-structure. Furthermore, given a formula $\phi$, we denote as
$\phi[\Sigma/\Sigma']$ the formula obtained by replacing any symbol from
$\Sigma$ with the corresponding primed symbol of $\Sigma'$. The notation
$\Sigma''$ \etc has similar meaning. 

We define first-order words as follows.

\begin{definition}[First-order word]
  \label{def:first-order:word}
  Let $\Sigma$ be a \emph{signature} and $\theory$ a $\Sigma$-theory. A
  \emph{first-order word} modulo $\theory$ (or $\theory$-word) is a finite
  sequence $\model=\seq{\letter_0,\ldots,\letter_{n-1}}$ of $\theory$-structures
  $\letter_i$ such that for each \emph{rigid} symbol $s\in\Sigma$, we have
  $s^{\letter_i}=s^{\letter_j}$ for all \mbox{$0 \le i,j < n$}.
\end{definition}

Intuitively, $\theory$-words are sequences of $\theory$-structures, where the
interpretation of \emph{rigid} symbols is constant throughout the word. As a
shortcut, then, only for \emph{rigid} symbols $s$, thus including sorts, we
denote as $s^{\bar\sigma}$ this unique interpretation. Since sort symbols are
always rigid, the domain of each sort is fixed throughout the word, so in the
terminology of first-order modal logics we are in a \emph{constant domains}
setting. First-order words are just called \emph{words} when there is no
ambiguity. We assume $\theory=\emptyset$ when not specified. The \emph{length}
of a word $\model=\seq{\letter_0,\ldots,\letter_{n-1}}$ is denoted as
$\abs{\model}=n$.

Given a $\Sigma$-theory $\theory$, the set of all the words and non-empty words
over $\sem{\theory}$ is then denoted as $\sem{\theory}^*$ and $\sem{\theory}^+$,
respectively, as usual.

\begin{definition}[First-order languages]
  For a signature $\Sigma$ and a $\Sigma$-theory~$\theory$, a \emph{first-order
  $\theory$-language} is a set $L\subseteq\sem{\theory}^*$.
\end{definition}

\medskip

\textbf{First-order linear temporal logic.} 
\FOLTL~\cite{GabbayKWZ03,KontchakovLWZ04} is a first-order temporal logic
that blends \LTL and first-order logic.  The syntax can be described as
follows. Let $\Sigma$ be a signature. A $\Sigma$-term is defined as
follows:
\begin{equation*}
  t \coloneqq c \mid x \mid f(t_1,\ldots,t_n)
\end{equation*}
where $c\in\Sigma$ is a constant, $x$ is a first-order variable (implicitly
associated to a sort from $\Sigma$), $f\in\Sigma$ is an $n$-ary function symbol,
and $t_1,\ldots,t_n$ are $\Sigma$-terms. The syntax of \FOLTL formulas is
defined as follows:
\begin{align*}
  \phi \coloneqq 
    p(t_1,\ldots,t_n)       \mid {} &
    t_1 = t_2               \mid
    \neg\phi                \mid
    \phi \lor \phi          \mid
    \exists x \phi          \\
    \mid {} & \ltl{X\phi}           \mid
    \ltl{Y\phi}             \mid
    \ltl{\phi U \phi}       \mid
    \ltl{\phi S \phi}
\end{align*}
where $t_1,\ldots,t_n$ are $\Sigma$-terms, $p$ is an $n$-ary predicate
symbol of $\Sigma$ (where $n$ can also be $0$, having atomic proposition),
and $x$ is a first-order variable. The temporal operators $\ltl{X,Y,U,S}$
are called \emph{tomorrow}, \emph{yesterday}, \emph{until}, and
\emph{since} respectively. We point out that classic first-order logic
formulas correspond to \FOLTL formulas $\phi$ without temporal operators.

$\Sigma$-terms are evaluated on $\Sigma$-structures. Given a signature $\Sigma$
and a $\Sigma$-structure $\sigma$, an \emph{environment} over $\sigma$ is a
mapping $\xi$ from any first-order variables of sort $S\in\Sigma$ to values of
their respective domain $S^\sigma$. Then, given a $\Sigma$-term $t$ and a
$\Sigma$-structure $\sigma$ and an environment $\xi$ over $\sigma$,
the \emph{evaluation} of $t$ on $\sigma$ with environment $\xi$, denoted
$t^{\sigma,\xi}$ is defined in a standard way, that is, $c^{\sigma,\xi} =
c^\sigma$, $x^{\sigma,\xi} = \xi(x)$, and $f(t_1,\ldots,t_n)^{\sigma,\xi} =
f^\sigma(t_1^{\sigma,\xi},\ldots,t_n^{\sigma,\xi})$, where $c\in\Sigma$ is a
constant, $x$ is a first-order variable, $f\in\Sigma$ is an $n$-ary function
symbol, and $t_1,\ldots,t_n$ are $\Sigma$-terms. 

Formulas of \FOLTL are interpreted over non-empty first-order words. Let $\phi$
be a \FOLTL formula and let $\bar\sigma
=\seq{\sigma_0,\ldots,\sigma_{\abs{\sigma}-1}}\in \sem{\theory}^+$. The
\emph{satisfaction} of $\phi$ by $\bar\sigma$, with environment $\xi$ over
$\sigma_0$, at time $i$, denoted $\bar\sigma,\xi,i\models\phi$, is defined as
follows:
\begin{itemize}
  \item $\bar\sigma,\xi,i\models p(t_1,\ldots,t_n)$ iff $(t_1^{\sigma_i,\xi},\ldots,t_n^{\sigma_i,\xi})\in p^{\sigma_i}$;
  \item $\bar\sigma,\xi,i\models t_1 = t_2$ iff
    $t_1^{\sigma_i,\xi}=t_2^{\sigma_i,\xi}$;
  \item $\bar\sigma,\xi,i\models\neg\phi$ iff $\bar\sigma,\xi,i\not\models\phi$;
  \item $\bar\sigma,\xi,i\models\phi_1\lor\phi_2$ iff 
    $\bar\sigma,\xi,i\models\phi_1$ or 
    $\bar\sigma,\xi,i\models\phi_2$;   
  \item $\bar\sigma,\xi,i\models\exists x\phi$ iff
    $\bar\sigma,\xi',i\models\phi$ \emph{for some} $\xi'$ agreeing with $\xi$
    except possibly for $x$;
  \item $\bar\sigma,\xi,i\models\ltl{X\phi}$ iff $i<|\bar\sigma|-1$ and 
    $\bar\sigma,\xi,(i+1)\models\phi$;
  \item $\bar\sigma,\xi,i\models\ltl{Y\phi}$ iff $i>0$ and 
    $\bar\sigma,\xi,(i-1)\models\phi$;
  \item $\bar\sigma,\xi,i\models\ltl{\phi_1 U \phi_2}$ iff there exists a $i\le
    k < \abs{\bar\sigma}$ such that $\bar\sigma,\xi,k\models\phi_2$ and
    $\bar\sigma,\xi,j\models\phi_1$ for all $i\le j < k$;
  \item $\bar\sigma,\xi,i\models\ltl{\phi_1 S \phi_2}$ iff there exists
    a $0 \le k\le i$ such that $\bar\sigma,\xi,k\models\phi_2$ and
    $\bar\sigma,\xi,j\models\phi_1$ for all $k < j \le i$.
\end{itemize}

We define the following shortcuts:
\begin{enumerate*}[label=(\roman*)]
  \item $\ltl{wX \phi == ! X ! \phi}$, called \emph{weak tomorrow};
  \item $\ltl{wY \phi == ! Y ! \phi}$, called \emph{weak yesterday};
  \item $\ltl{\phi_1 R \phi_2 == !(!\phi_1 U !\phi_2)}$, called
    \emph{release};
  \item $\ltl{\phi_1 T \phi_2 == !(!\phi_1 S !\phi_2)}$, called
    \emph{triggered}.
\end{enumerate*}
We point out that, for any \FOLTL formula $\phi$, $\ltl{wX\phi}$ (resp.,
$\ltl{wY \phi}$) is true at the last (resp., first) time point of
a first-order word. The operators $\ltl{X,wX,U,R}$ are called
\emph{future} temporal operators, while $\ltl{Y,wY,S,T}$ are called
\emph{past} temporal operators.

Given a \FOLTL \emph{$\Sigma$-sentence} $\phi$ (\ie a formula with no free
variables), we say that $\bar\sigma$ \emph{satisfies} $\phi$, denoted
$\bar\sigma\models\phi$, if $\bar\sigma,\xi,0\models\phi$ for all $\xi$.
Note that \FOLTL semantics subsumes that of classical first-order logic.
A first-order sentence $\phi$ is satisfied by a $\Sigma$-structure $\sigma$
iff $\bar\sigma\models\phi$ where $\bar\sigma=\seq{\sigma}$, and this
corresponds to the classical first-order semantics.

\begin{definition}[Language denoted by a sentence]
  \label{def:denoted:language}
  Given a signature $\Sigma$ and a $\Sigma$-theory $\theory$, we denote as
  $\lang_\theory(\phi)$ the $\theory$-language $L\subseteq\sem{\theory}$ such
  that $\bar\sigma\models\phi$ if and only if $\bar\sigma\in L$.
\end{definition}

A \FOLTL $\Sigma$-sentence $\phi$ is \emph{$\theory$-satisfiable} if and only if
$\lang_\theory(\phi) \ne \emptyset$. \FOLTL satisfiability, \ie the problem of
establishing whether a \FOLTL formula is satisfiable, is known to be
undecidable, in general. See \cite{GabbayKWZ03} for an extensive treatment
where, however, only monosorted signatures are considered, and constants and
relational symbols are always assumed to be rigid and non-rigid, respectively.

%% file: sections/4.automata.tex

\section{First-Order Automata}
\label{sec:automata}

A \emph{first-order automaton} is a device that accepts first-order languages. A
naive definition would result into an infinite-state representation not suitable
for being handled computationally. Instead, here we make use of a
\emph{symbolic} representation, which exploits first-order logic to represent
such infinite objects finitarily.

\begin{definition}[First-order automata]
  \label{def:automata}
  A \emph{first-order automaton} is a tuple
  $\autom=\seq{\Sigma,\Gamma,\phi_0,\phi_T, \phi_F}$ where:
  \begin{itemize}
    \item $\Sigma$ is a \emph{finite} signature called the \emph{word
    signature};
    \item $\Gamma$ is a \emph{finite} signature, \emph{disjoint} from $\Sigma$
      except for having the same \emph{sort} symbols, called the \emph{state
      signature};
    \item $\phi_0$ is a first-order \emph{$\Gamma$-sentence}
      called the \emph{initial condition};
    \item $\phi_T$ is a first-order \emph{$(\Gamma\cup\Sigma\cup\Gamma')$-sentence}
      called the \emph{transition relation};
    \item $\phi_F$ is a first-order \emph{$\Gamma$-sentence} called the \emph{acceptance
      condition}.
  \end{itemize}
\end{definition}

Intuitively, an automaton has an infinite state space made of
$\Gamma$-structures, which are called \emph{states}, with a transition relation
represented symbolically by $\phi_T$. We say only \emph{automaton} to mean a
first-order automaton as in \cref{def:automata}, when there is no ambiguity.
Automata induce \emph{runs}, \ie sequences of states, when \emph{reading} a
$\Sigma$-word.

\begin{definition}[Run]
  \label{def:run}
  Let  $\autom=\seq{\Sigma,\Gamma,\phi_0,\phi_T, \phi_F}$ be an automaton, and
  let $\model=\seq{\letter_0,\ldots,\letter_{n-1}}$ be a possibly empty
  $\theory$-word for some $\Sigma$-theory $\theory$. A \emph{run} of $\autom$
  induced by $\model$ is an $\emptyset$-word
  $\bar\rho=\seq{\rho_0,\ldots,\rho_n}$ over $\Gamma$ such that:
  \begin{enumerate}
    \item if $n>0$, then $S^{\bar\rho}=S^{\bar\sigma}$ for all
    $S\in\sorts(\Sigma)$;
    \item $\rho_0\models\phi_0$; and
    \item $\rho_i\oplus\letter_i\oplus\rho_{i+1}'\models \phi_T$ for all $0\le i < n$.
  \end{enumerate}
  where $\rho_i\oplus\letter_i\oplus\rho_{i+1}'$ is the common
  $(\Gamma\cup\Sigma\cup\Gamma')$-expansion of $\rho_i$, $\letter_i$ and
  $\rho_{i+1}'$.
\end{definition}

Note that automata are \emph{nondeterministic}, hence the same word induces more
than one run, in general. The acceptance condition of first-order automata is a
standard reachability condition, defined as follows.
\begin{definition}[Acceptance of words]
  \label{def:acceptance}
  Given an automaton $\autom=\seq{\Sigma,\Gamma,\phi_0,\phi_T, \phi_F}$, let
  $\theory$ be a $\Sigma$-theory, and let
  $\model=\seq{\letter_0,\ldots,\letter_{n-1}}$ be a $\theory$-word. We say that
  $\autom$ \emph{accepts} $\model$ iff there exists a \emph{run}
  $\bar\rho=\seq{\rho_0,\ldots,\rho_n}$ induced by $\bar\sigma$ such that
  $\rho_n\models\phi_F$.
\end{definition}

The \emph{$\theory$-language} of an automaton $\autom$, denoted
$\lang_\theory(\autom)$, is the set $L\subseteq\sem{\theory}^*$ of all the
$\theory$-words accepted by $\autom$. Then, \emph{first-order regular languages}
are defined consequently.

\begin{definition}[First-order $\theory$-regular language]
  \label{def:regular}
  Given a $\Sigma$-theory $\theory$, a $\theory$-language
  $L\subseteq\sem{\theory}^*$ is said to be \emph{$\theory$-regular} iff there
  is a first-order automaton $\autom$ such that $\lang_\theory(\autom)=L$.
\end{definition}

\begin{example}
  Consider a signature $\Sigma=\set{R, c, >,+,\cdot,0,1}$ where $R$ is a sort,
  $c$ is a non-rigid constant and $>$, $+$, $\cdot$, and $0$ are the standard
  (rigid) arithmetic symbols. Let $\theory$ be the theory of \emph{real closed
  fields}\footnotemark, over the signature expanded with the constant $c$, with
  no additional axioms constraining $c$. Let
  $\autom=\seq{\Sigma,\Gamma,\phi_0,\phi_T, \phi_F}$ where
  \begin{enumerate}
    \item $\Gamma=\set{R, a,n}$ with $a$ and $n$ two constants of sort $R$;
    \item $\phi_0\coloneqq n = 0$;
    \item $\phi_T\coloneqq n' = n + 1 \land (n + 1)\cdot a' = n \cdot a + c$;
    and 
    \item $\phi_F\coloneqq a = 1$.
  \end{enumerate}
  It can be seen that $\autom$ accepts all the $\theory$-words such that the
  average of the values of $c$ over all the word is one.
\end{example}

\footnotetext{%
  The \emph{theory of real closed fields} is the first-order theory of $(\mathbb
  R,+,\cdot,<,0,1)$. This theory does not characterize $\mathbb R$ up to
  isomorphism; for instance, the field of real algebraic numbers is a
  (countable) model of the same first-order theory.%
}

\subsection{Basic properties of $\theory$-regularity.}
Let us now study some basic questions about $\theory$-regular languages.
As a first step, some cardinality arguments allow us to conclude that, despite
the evident expressive power of the formalism, there are $\theory$-languages,
for some $\theory$, that cannot be accepted by any first-order automaton.

\begin{restatable}{theorem}{cardinalitythm}
  \label{thm:cardinality}
  For some signature $\Sigma$ and $\Sigma$-theory $\theory$, there are
  $\theory$-languages that are not $\theory$-regular.
\end{restatable}

Despite \cref{thm:cardinality}, a non-trivial example of non $\theory$-regular
language has still not been found. Now, we can show that $\theory$-regular
languages are closed under some common operations. We start with union and
intersection.

\begin{restatable}{theorem}{unionintersection}
  For any theory $\theory$, $\theory$-regular languages are closed under union
  and intersection.
\end{restatable}

We look now at concatenation and Kleene star. The following definition
accounts for the fact that the simple juxtaposition of two $\theory$-words may
not be a $\theory$-word because of different interpretations of the rigid
symbols.
\begin{definition}[Concatenation of $\theory$-languages]
  For some signature $\Sigma$ and $\Sigma$-theory $\theory$, let
  $L_1,L_2\subseteq\sem{\theory}^*$ be two $\theory$-languages. The
  \emph{concatenation} of $L_1$ and $L_2$ is defined as:
  \begin{equation*}
    L_1\cdot L_2 = \set{uv \mid u\in L_1, v\in L_2, \text{$uv$ is a $\theory$-word}}
  \end{equation*}
\end{definition}

\begin{restatable}{theorem}{thmconcatenation}
  \label{thm:concatenation}
  $\theory$-regular languages are closed under concatenation, for any $\theory$.
\end{restatable}

We define the closure under Kleene star of $\theory$-languages as follows.
\begin{definition}[Kleene star of a $\theory$-language]
  For some signature $\Sigma$ and $\Sigma$-theory $\theory$, let
  $L\subseteq\sem{\theory}^*$ be a $\theory$-language. The \emph{Kleene star}
  of $L$ is defined as:
  \begin{equation*}
    L^*=\set{\varepsilon}\cup\set*{u_0\cdots u_n \middle| 
      \begin{aligned}
        &\text{$u_i\in L$ for $0\le i \le n$, $n\in\mathbb{N}$}\\
        &\text{$u_0\cdots u_n$ is a $\theory$-word}
      \end{aligned}
    }
  \end{equation*}
\end{definition}

To build first-order automata accepting the Kleene star of
$\theory$-regular languages, we need to investigate more closely the
expressive power of first-order automata. It turns out they leak towards
second-order territory. 

\begin{example}
  \label{ex:reachability}
  Consider the signature $\Sigma=\set{R,s,d}$ with $R$ a binary predicate and
  $s$ and $d$ two constants. It is well-known that one cannot write a
  first-order formula $\phi$ that holds in a $\Sigma$-structure if and only if
  $d$ is \emph{not} $R$-reachable from $s$. However, consider the language
  $\lang$ made of $\emptyset$-words $\bar\sigma$ over $\Sigma$ where each
  element $\sigma$ of the word is such that $d^\sigma$ is \emph{not}
  $R^\sigma$-reachable from $s^\sigma$. We can build a first-order automaton
  $\autom=\seq{\Sigma,\Gamma,\phi_0,\phi_T,\phi_F}$ that accepts $\lang$, as
  follows. We define $\Gamma=\set{P}$ where $P$ is a unary predicate, and
  $\phi_0\equiv\phi_f\equiv\top$. Then, we define $\phi_T$ as follows:
  \begin{equation*}
    \phi_T \equiv P(s) \land \forall x y ((P(x) \land R(x,y)) \to P(y)) \land \neg P(d)
  \end{equation*}
  Intuitively, the additional predicate $P$ provided by $\Gamma$ acts as the
  fixpoint of $R$. Indeed, the second-order logic formula $\exists P\suchdot
  \phi_T$ is exactly the one expressing non-reachability between $s$ and $d$ in
  $\Sigma$.
\end{example}

\Cref{ex:reachability} shows that there are non-first-order-definable properties
that can be enforced on single letters of $\theory$-words. In general, these are
\emph{existential} second-order properties.
\begin{restatable}{theorem}{thmsecondorderclass}
  \label{thm:second-order:class}
  Let $\Sigma$ be a first-order signature, and $\theory$ a $\Sigma$-theory. Let
  $\mathcal{C}$ be a class of $\Sigma$-structures, and let
  $L_{\mathcal{C}}\subseteq\sem{\theory}^*$ be the language of all the
  $\theory$-words whose letters belong to $\mathcal{C}$. If $\mathcal{C}$ is
  definable by an \emph{existential second-order logic} $\Sigma$-sentence, then
  $L_{\mathcal{C}}$ is $\theory$-regular.
\end{restatable}

Intuitively, the existence of the run, required for accepting the word, performs
the role of the existential second-order quantifiers. Note that this means that
one may directly use existential second-order sentences as transition relations.
\begin{definition}[$\Sigma_1^1$-automata]
  An \emph{$\Sigma_1^1$-automaton} is a tuple
  $\autom_1=\seq{\Sigma,\Gamma,\phi_0,\phi_T,\phi_F}$ similar to a first-order
  automaton except that $\phi_0$, $\phi_T$, and $\phi_F$ are \emph{existential
  second-order} sentences. Acceptance of words resembles
  \cref{def:acceptance}.
\end{definition}
\begin{restatable}{lemma}{secondorderautomata}
  \label{thm:second-order:automata}
  $\Sigma_1^1$-automata are equivalent to first-order ones.
\end{restatable}

The syntactic sugar of $\Sigma_1^1$-automata helps us prove the closure
under Kleene star of $\theory$-regular languages.
\begin{restatable}{theorem}{thmkleene}
  \label{thm:kleene}
  $\theory$-regular languages are closed under Kleene star, for any $\theory$.
\end{restatable}

Therefore, $\theory$-regular languages are closed under the usual
language-theoretic operations of regular expressions, which justifies the term
\emph{regular}.
Whether they are closed under complementation as well is still
an open question. The usual path is complementing an automaton by first
determinizing it, therefore a preliminary step for complementation is that of
defining determinism and deterministic automata.

\subsection{Determinism}
Determinism is more difficult to define in symbolically represented
automata such as first-order ones, because the transition relation is only
implicitly presented by $\phi_T$. 
In principle, to define what it means for a first-order automaton to be
\emph{deterministic}, one would force it to have only one initial state and
the transition relation to be a function, that is, when fixed a source
state and a letter, the destination state would be unique.  The purpose is
ensuring that for any given word there is only one run of the automaton. 

This brings us to the notion of \emph{implicit definability}, which has been
introduced originally by \citet{Beth53}, extended by \citet{Craig57}, and has
been subject of extensive research.

\begin{definition}[Implicit definability]
\label{def:definability}
  Let $\Sigma$ be a signature and $\phi$ a $\Sigma$-sentence. Let
  $\Delta,\Lambda\subseteq\Sigma$ be such that:
  \begin{itemize}
    \item $\Sigma = \Delta\cup\Lambda$;
    \item $\sorts(\Sigma)\subseteq\Lambda$;
    \item $\Delta \cap \Lambda = \emptyset$.
  \end{itemize}
  We say that $\Delta$ is \emph{implicitly defined by $\Lambda$ under
  $\phi$} if every $\Lambda$-structure has at most one $\Sigma$-expansion
  satisfying $\phi$.
\end{definition}
Equivalently, $\Delta$ is implicitly defined by $\Lambda$
under $\phi$ if any two models of $\phi$ that agree on
the symbols of $\Lambda$ also agree on the symbols of
$\Delta$. Notice that all sort symbols belong to $\Lambda$:
hence the domains are fixed by the $\Lambda$-structure,
while implicit definability concerns the interpretation of the
non-sort symbols in $\Delta$.

\begin{definition}[Deterministic first-order automaton]
\label{def:determinism}
  Let $\autom=\seq{\Sigma,\Gamma,\phi_0,\phi_T,\phi_F}$ be a first-order
  automaton. We say that $\autom$ is \emph{deterministic} if the following
  conditions hold.
  \begin{enumerate}
    \item For the initial condition, let $\Delta_0 \coloneqq
      \Gamma\setminus\sorts(\Gamma)$ and $\Lambda_0 \coloneqq
      \sorts(\Gamma)$. Then $\Delta_0$ is implicitly defined by
      $\Lambda_0$ under $\phi_0$.
    \item For the transition relation, let $\Lambda_T \coloneqq
      \Gamma\cup\Sigma$ and let $\Delta_T \coloneqq \set{\,s' \mid
      s\in\Gamma\setminus\rigid(\Gamma)\,}$ be the set of primed non-rigid
      symbols in $\Gamma$. Then $\Delta_T$ is implicitly defined by $\Lambda_T$
      under $\phi_T$.
    \item The $\phi_F$ sentence has a uniform truth value over initial
      states, \ie either $\phi_0 \to \phi_F$ or $\phi_0 \to \lnot\phi_F$.
  \end{enumerate}
\end{definition}

Intuitively, Item~1 requires the interpretation of all the non-sort symbols
of an initial state to be uniquely determined by the domains of the sorts,
which are fixed by the input $\theory$-word. Thus, for any fixed choice of
domains, there is at most one state satisfying $\phi_0$.
Item~2 requires the interpretation of the non-rigid symbols of the
destination state to be uniquely determined once the current state and the
current input letter are fixed. Thus, for every state and input letter,
there is at most one successor satisfying $\phi_T$.
Finally, Item~3 ensures that all possible initial states agree on the
acceptance of the empty word. This is needed because, unlike nonempty
words, the empty word does not fix the interpretation of the sort symbols
and may therefore induce multiple initial states.

\subsection{Completeness}
A feature that will often come useful together with determinism is
\emph{completeness}, that is, the guarantee of existence of at least one run for
any word. To ensure this we need to ask the initial condition to be satisfiable
by models that interpret sorts in whatever way the word does, and the transition
function to guarantee a successor for any state and any letter.
\begin{definition}[Complete initial state condition]
\label{def:complete:phi0}
  A $\Gamma$-sentence $\phi_0$ is \emph{complete} if each
  $\sorts(\Gamma)$-structure has \emph{at least} one $\Gamma$-expansion $\sigma$
  such that $\sigma\models\phi_0$.
\end{definition}
\begin{definition}[Total transition function]
\label{def:total:trans:fun}
  A $(\Gamma\cup\Sigma\cup\Gamma')$-sentence $\phi_T$ is \emph{total} if each
  $(\Gamma\cup\Sigma)$-structure has \emph{at least} one
  $(\Gamma\cup\Sigma\cup\Gamma')$-expansion $\sigma$ such that
  $\sigma\models\phi_T$.
\end{definition}

\begin{definition}[Complete automata]
  \label{def:completeness}
  An automaton $\autom=\seq{\Sigma,\Gamma,\phi_0,\phi_T,\phi_F}$ is
  \emph{complete} if $\phi_0$ is \emph{complete} and $\phi_T$ is \emph{total}.
\end{definition}

In finite-state automata, a straightforward consequence of determinism and
completeness is that the run of a given word exists and is unique. Here, from
\cref{def:determinism,def:completeness} we obtain a similar situation.
\begin{restatable}{proposition}{detuniquerun}
\label{prop:determinism:unique:run}
  For any \emph{complete} and \emph{deterministic} first-order automaton
  $\autom$ and any nonempty $\theory$-word $\bar\sigma$ the run of
  $\bar\sigma$ over $\autom$ exists and is unique. Moreover, the empty word
  induces at least one run, and either all its runs are accepting or none
  of them is.
\end{restatable}

Therefore, the next proposition shows that, when determinization is
possible, $\theory$-regular languages are closed under complementation.

\begin{restatable}{proposition}{detcomplementation}
  \label{prop:deterministic:complementation}
  From a \emph{complete} and \emph{deterministic} first-order automaton
  $\autom$ one can obtain a complement automaton $\overline{\autom}$ such that
  $\lang_\theory(\overline{\autom})=\overline{\lang_\theory(\autom)}$ for any theory
  $\theory$.
\end{restatable}

%% file: sections/5.encoding.tex

\section{From Temporal Logic to Automata}
\label{sec:encoding}

We can now state how to encode \FOLTL $\Sigma$-sentences into
language-equivalent first-order automata. In all that follows, \emph{w.l.o.g.}\
we assume any sentence to be in \emph{negated normal form}, \ie with negations
only applied to atoms as in $\neg p(t_1,\ldots,t_n)$ or $\neg(t_1 = t_n)$.

The translation first involves the following normal form.
\begin{definition}[Stepped normal form]
  \label{def:snf}
  Given an \FOLTL formula $\phi$, the \emph{stepped normal form} of $\phi$,
  denoted $\snf(\phi)$, is the formula defined recursively as follows:
  \begin{enumerate}
    \item $\snf(p(t_1,\ldots,t_n))=p(t_1,\ldots,t_n)$ and $\snf(t_1 = t_2)=(t_1
    = t_2)$;
    \item $\snf(Qx\phi)= Qx\snf(\phi)$, where $Q\in\set{\forall,\exists}$;
    \item $\snf(\neg\phi)=\neg\snf(\phi)$;
    \item $\snf(\phi_1\circ\phi_2)=\snf(\phi_1)\circ\snf(\phi_2)$, where
    $\circ\in\set{\land,\lor}$
    \item $\snf(\circ \phi)=\circ\phi$ where
      $\circ\in\set{\ltl{X,Y,wX,Z}}$;
    \item $\snf(\ltl{\phi_1 U \phi_2})=\snf(\phi_2)\lor(\snf(\phi_1)\land\ltl{X(\phi_1 U \phi_2)})$;
    \item $\snf(\ltl{\phi_1 R \phi_2})=\snf(\phi_2)\land(\snf(\phi_1)\lor\ltl{wX(\phi_1 R \phi_2)})$;
    \item $\snf(\ltl{\phi_1 S \phi_2})=\snf(\phi_2)\lor(\snf(\phi_1)\land\ltl{Y(\phi_1 S \phi_2)})$;
    \item $\snf(\ltl{\phi_1 T \phi_2})=\snf(\phi_2)\land(\snf(\phi_1)\lor\ltl{Z(\phi_1 T \phi_2)})$.
  \end{enumerate}
\end{definition}

Any temporal operator in $\snf(\phi)$ appears only directly below one among
$\ltl{X,Y,wX,Z}$. Intuitively, $\snf(\phi)$ separates what the formula says
about the current state, from what it says about the next and previous
ones. The following defines a universe of formulas that are important for
the construction of a first-order automaton.
\begin{definition}[Closure]
  \label{def:closure}
  The \emph{closure} of an \FOLTL $\Sigma$-sentence $\phi$ is the set
  $\closure(\phi)$ defined as follows:
  \begin{enumerate}
    \item $\ltl{X\phi}\in\closure(\phi)$;
    \item $\psi\in\closure(\phi)$ for any subformula $\psi$ of $\phi$
      (including itself);
    \item for any $\ltl{\phi_1 U \phi_2}\in\closure(\phi)$, $\ltl{\phi_1 R
    \phi_2}\in\closure(\phi)$, $\ltl{\phi_1 S \phi_2}\in\closure(\phi)$, or
    $\ltl{\phi_1 T \phi_2}\in\closure(\phi)$, we have $\ltl{X(\phi_1 U
    \phi_2)}\in\closure(\phi)$, $\ltl{wX(\phi_1 R \phi_2)}\in\closure(\phi)$,
    $\ltl{Y(\phi_1 S \phi_2)}\in\closure(\phi)$, or $\ltl{Z(\phi_1 T
    \phi_2)}\in\closure(\phi)$, respectively.
  \end{enumerate}
\end{definition}

\subsection{Surrogates}
Now, we define some symbols, called \emph{surrogates}, that take the place
of temporal formulas in the encoding below.
\begin{gather*}
  \XS=\set{\xs_\psi \mid \ltl{X\psi}\in\closure(\phi) } \quad
  \wXS=\set{\ws_\psi \mid \ltl{wX\psi}\in\closure(\phi) } \\
  \YS=\set{\ys_\psi \mid \ltl{Y\psi}\in\closure(\phi) } \quad
  \ZS=\set{\zs_\psi \mid \ltl{Z\psi}\in\closure(\phi) }
\end{gather*}
where the $\xs_\psi$, $\ws_\psi$, $\ys_\psi$, $\zs_\psi$ are predicates of the
arity $n$ corresponding to the number of free first-order variables in $\psi$.

Intuitively, the surrogate predicates act as auxiliary memory, whose
interpretations are ``stored in the states'' of the automaton. Their
intended meaning relies on distinguishing the states of a run from the
positions of the input word. Consider an accepting run $\bar\rho
= \seq{\rho_0,\ldots,\rho_n}$ of the automaton to be defined below over
a word $\bar\sigma = \seq{ \sigma_0,\ldots,\sigma_{n-1}}$.  For $0 \leq
i < n$, the state $\rho_i$ precedes the reading of $\sigma_i$, which labels
the transition from $\rho_i$ to $\rho_{i+1}$. In particular, the run starts
in $\rho_0$ and reads $\sigma_0$ on its first transition.

Past surrogates record information about the preceding input position.  For
example, for $0 < i \leq n$ and a fixed environment $\xi$, the truth of
$\ys_{\psi}(\bar x)$ in $\rho_i$ records that $\psi(\bar x)$ holds at
position $i-1$ of $\bar\sigma$ under $\xi$.  Future surrogates instead
represent obligations for the next transition to be taken: for $0 \leq
i < n$, the truth of $\xs_{\psi}(\bar x)$ in $\rho_i$ requires $\psi(\bar
x)$ to hold at position $i$ of $\bar\sigma$, whose letter $\sigma_i$ is
read on that transition. Thus, the obligation stored by $\xs_{\psi}$ in
$\rho_i$ concerns input position $i$, not $i+1$.  The weak surrogates have
the same roles, but differ in their interpretation at the boundaries of the
word.

For every $\psi \in \closure(\phi)$, we denote with $\snf_S(\psi)$ the
\emph{first-order formula} obtained from $\snf(\psi)$ by replacing each
temporal subformula $\ltl{X\theta(\bar x)}$, $\ltl{wX\theta(\bar x)}$,
$\ltl{Y\theta(\bar x)}$, or $\ltl{Z\theta(\bar x)}$ with
$\xs'_{\theta}(\bar x)$, $\ws'_{\theta}(\bar x)$, $\ys_{\theta}(\bar x)$,
or $\zs_{\theta}(\bar x)$, respectively, where $\bar x$ is a finite tuple
of variables containing the free variables of $\theta$.  It is worth noting
that, by definition, $\snf_S(\psi)$ is a first-order formula over the
signature $\Sigma$ expanded with the surrogate predicates and their primed
version.

In the encoding below, these surrogate predicates, together with the sort
symbols of $\Sigma$, form the state signature $\Gamma$. When reading
$\sigma_i$, the formula $\operatorname{snf}_S(\psi)$ is evaluated in
$\rho_i \oplus \sigma_i \oplus \rho'_{i+1}$.  Thus, symbols from $\Sigma$
are interpreted in the current input letter, unprimed surrogates in the
source state $\rho_i$, and primed surrogates in the destination state
$\rho_{i+1}$.

To see why future surrogates are primed, consider an obligation
$\xs_{\ltl{p(x) U q(x)}}(x)$ holding in $\rho_i$ under an environment
$\xi$. The transition reading $\sigma_i$ must ensure that either $q(x)$
holds in $\sigma_i$, or $p(x)$ holds in $\sigma_i$ and the same until
obligation is carried to $\rho_{i+1}$, in each case under the same
environment. Accordingly, the translation is $\snf_S(\ltl{p(x) U q(x)})
= q(x) \lor (p(x) \land \xs'_{\ltl{p(x) U q(x)}}(x))$. The prime is
essential here: it places the continuation obligation in $\rho_{i+1}$,
where it concerns the next input position, rather than in $\rho_i$, where
it would concern the current one.

By contrast, evaluating $\ltl{Y}\theta(\bar x)$ at input position $i$
requires information about position $i-1$. This information is already
recorded by $\ys_{\theta}(\bar x)$ in the source state $\rho_i$, so no
prime is needed. The same distinction applies to the weak operators.

\subsection{The encoding}
The following encoding formalizes the construction of a first-order
automaton starting from an \FOLTL formula.

\begin{encoding}[From \FOLTL to automata]
\label{enc:encoding}
  Let $\phi$ be an \FOLTL $\Sigma$-sentence. We define
  $\autom(\phi)=\seq{\Sigma,\Gamma,\phi_0,\phi_T,\phi_F}$ as follows:
  \begin{enumerate}
    \item the state signature is $\Gamma=\sorts(\Sigma) \cup
      \XS\cup\wXS\cup\YS\cup\ZS$;
    \item the initial and final conditions are:
    \begin{align*}
      \phi_0\coloneqq {} & \xs_\phi \land
        \smashoperator[r]{\bigwedge_{\zs_\psi\in\ZS}}
          \forall \bar x. \zs_\psi(\bar x) \land
        \smashoperator[r]{\bigwedge_{\ys_\psi\in\YS}}
          \forall \bar x. \neg \ys_\psi(\bar x)\\
      \phi_F\coloneqq {} &
        \smashoperator{\bigwedge_{\ws_\psi\in\wXS}}
          \forall \bar x. \ws_\psi(\bar x) \land
        \smashoperator[r]{\bigwedge_{\xs_\psi\in\XS}}
          \forall \bar x. \neg \xs_\psi(\bar x)
    \end{align*}
    \item the transition relation is:
    \begin{align}
      \phi_T  \coloneqq {} &
       \bigwedge_{s_\psi \in \XS\cup\wXS}
        \forall \bar x.
          [s_\psi(\bar x) \iff
              \snf_S(\psi(\bar x))
          ]
          \label{eq:encoding:one}\\
        {} \land {} & \bigwedge_{s_\psi \in \YS\cup\ZS}
        \forall \bar x.
          [s_\psi'(\bar x) \iff
              \snf_S(\psi(\bar x))
          ]
          \label{eq:encoding:two}
    \end{align}
  \end{enumerate}
\end{encoding}
Intuitively, the initial condition states that, in the first position of the
word, $\phi$ must hold and all formulas of type $\ltl{Y}\psi$ (resp., $\ltl{wY}
\psi$) in the closure are false (resp., true), following \FOLTL semantics.
  
It is worth noting two aspects of the initial condition.  First, the
conjunct $\xs_\phi$ is used to enforce the satisfaction of $\phi$ at the
first position of the input word. Indeed, $\xs_\phi$ is the surrogate
associated with $\ltl{X \phi}$; since the initial automaton state precedes
the first input letter, requiring $\xs_\phi$ in the initial state, together
with~\cref{eq:encoding:one} of the transition relation, forces
$\snf_S(\phi)$ to hold along the first transition, and thus $\phi$ to hold
at the first position of the word.  Second, since $\phi$ is a \FOLTL
sentence and therefore has no free first-order variables, the surrogate
$\xs_\phi$ has arity zero.

Then, the transition relation ensures that the surrogates of all the future
and past formulas behave following the semantics of such formulas. The
final condition ensures we are at the end of a finite word, by checking
whether all formulas of type $\ltl{X}\psi$ (resp., $\ltl{wX} \psi$) in the
closure are false (resp., true). We can state and prove the correctness of
the above construction.
\begin{restatable}{theorem}{encodingproof}
  \label{thm:encoding}
  For any \FOLTL $\Sigma$-sentence $\phi$ and any $\Sigma$-theory
  $\theory$, we have $\lang_\theory(\autom(\phi))=\lang_\theory(\phi)$.
\end{restatable}

%% file: sections/6.purepast.tex

\section{Deterministic Automata From\\Pure-Past \FOLTL}
\label{sec:logic:purepast}

In this section, we consider the fragment of \FOLTL in which temporal operators
are restricted to talk about the past. We show how to go \emph{directly} from
formulas of this fragment to \emph{deterministic} first-order automata and we
discuss the importance of such construction in practical settings.

\subsection{The pure-past fragment of \FOLTL}

A \emph{pure-past} \FOLTL formula is a formula $\phi$ which does \emph{not}
use future temporal operators ($\ltl{X}$, $\ltl{wX}$, 
$\ltl{U}$, $\ltl{R}$). Pure-past \FOLTL is denoted as \FOpLTL. Given that
only past operators are available, it is natural to interpret  a \FOpLTL
formula at the \emph{end} of a word, as is customary in the pure-past
propositional
setting~\cite{de2021pure,DBLP:conf/time/ArtaleGGMM23,DBLP:journals/iandc/ArtaleGGMM25}.
Hence given a $\theory$-word $\bar\sigma$ and a \FOpLTL formula $\phi$, we
say that $\bar\sigma=\seq{\sigma_0,\ldots,\sigma_{n-1}}$ satisfies $\phi$,
denoted $\bar\sigma\models\phi$, iff $\bar\sigma,n-1\models\phi$ according
to the semantics of \FOLTL defined in~\cref{sec:languages}.
$\lang_\theory(\phi)$ is the set of $\theory$-words satisfying $\phi$.
The construction of \cref{enc:encoding} can be tailored to \FOpLTL formulas. 

\begin{encoding}[Automaton of a \FOpLTL sentence]
\label{enc:past:encoding}
  Let $\phi$ be a \FOpLTL $\Sigma$-sentence. We define
  $\autom(\phi)=\seq{\Sigma,\Gamma,\phi_0,\phi_T,\phi_F}$ as
  in~\cref{enc:encoding} with the following changes:
  \begin{enumerate}
    \item $\closure(\phi)$ is defined as in \cref{def:closure} but instead
      of $\ltl{X\phi}\in\closure(\phi)$ we require that
      $\ltl{Y\phi}\in\closure(\phi)$;
    \item we remove the conjunct $\xs_\phi$ from the initial condition
      $\phi_0$;
    \item we add the conjunct $\ys_\phi$ to the final condition $\phi_F$.
  \end{enumerate}
\end{encoding}

It is worth noticing that, by construction, the $\XS$ and $\wXS$ sets for
automata constructed starting from \FOpLTL formulas are empty. The
correctness of~\cref{enc:past:encoding} is stated in the following theorem.

\begin{restatable}{theorem}{encodingpastproof}
\label{thm:encodingpastproof}
  For every \FOpLTL $\Sigma$-sentence $\phi$ and every $\Sigma$-theory
  $\theory$, the automaton $\autom(\phi)$ obtained
  by~\cref{enc:past:encoding} satisfies $\lang_\theory(\autom(\phi))
  = \lang_\theory(\phi)$.
\end{restatable}

Crucially, the automaton resulting from this encoding of a pure-past
sentence is \emph{deterministic}. Intuitively, this is because $\phi_0$
completely determines the interpretation of all the predicates, and for
$\phi_T$, having only past operators, only \eqref{eq:encoding:two} of
\cref{enc:encoding} survives, which explicitly defines each predicate in
$\Gamma'$ in terms of the predicates in $\Gamma$ and $\Sigma$. It also
happens to be \emph{complete}, again because of the explicitly defining
form of \eqref{eq:encoding:two} of \cref{enc:encoding}.

\begin{restatable}{theorem}{pastdeterminism}
  \label{thm:purepast:determinism}
  Given a \FOpLTL $\Sigma$-sentence $\phi$, $\autom(\phi)$ is
  \emph{complete} and \emph{deterministic}.
\end{restatable}

\subsection{The relevance of the pure-past fragment}

Pure-past fragments of temporal logics~\cite{lichtenstein1985glory} are
gaining renewed interest for their practical use in formal verification,
such as in \emph{reactive synthesis} and \emph{monitoring}.

\emph{Reactive synthesis}~\cite{PnueliR89} involves creating a system that
satisfies a temporal specification. For a temporal formula $\phi$ with
\emph{controllable} ($\mathcal{C}$) and \emph{uncontrollable} ($\mathcal{U}$)
variables, it seeks to synthesize a controller that ensures $\phi$ is satisfied
regardless of $\mathcal{U}$.  For \LTL, this involves generating a
\emph{deterministic} automaton $\autom_\phi$ and solving fixpoint computations.
Generating such automata is complex and typically requires \emph{doubly}
exponential time. However, for pure-past \LTL formulas, deterministic finite
automata can be built directly in \emph{singly} exponential time, making the
problem only \EXPTIME-complete \cite{ArtaleGGMM23}. Symbolic techniques are also effective here,
allowing efficient representation of these automata.

\emph{Monitoring}~\cite{LeuckerS09} is a runtime verification technique
that detects violations of temporal properties during system execution. To
ensure irrevocable verdicts, it relies on deterministic automata, which are
challenging to produce from temporal specifications \cite{KupfermanV99}.
The use of pure-past \FOLTL is seen as promising for monitoring in
infinite-state settings, with deterministic automata being essential to
this approach.

%% file: sections/7.emptiness.tex

\section{Checking Emptiness}
\label{sec:emptiness}

Given \cref{thm:encoding}, we know that the non-emptiness problem for
first-order automata is highly undecidable, since so is \FOLTL satisfiability,
in general~\cite{GabbayKWZ03}. However, we can show how to at least recover
semi-decidability in some cases, setting the stage for future algorithmic
developments. To be more specific, what can be semi-decided, in some cases, is
the \emph{non-emptiness} problem of first-order automata, which translates to
the \emph{satisfiability} of a certain class of \FOLTL formulas.

For $k\ge0$, given a signature $\Omega$, an $\Omega$-sentence $\psi$, an
$\Omega$-theory $\mathcal{S}$, and an $\Omega$-structure $\sigma$, we write
$\Omega^k$, $\psi^k$, $\mathcal{S}^k$, and $\sigma^k$, respectively, to denote
the objects obtained by \emph{renaming} any \emph{non-rigid} symbol $s$ by
$s^k$, and any \emph{primed non-rigid} symbol $s'$ by $s^{k+1}$.

Now, given an automaton $\autom=\seq{\Sigma,\Gamma,\phi_0,\phi_T,\phi_F}$
and a $\Sigma$-theory $\theory$ for the words in its language, we define,
for any $k\ge0$, the formulas 
\begin{align*}
  \sem{\autom}_k   &\equiv \phi^0_0 \land \bigwedge_{i=0}^{k-1} \phi_T^i \\
  \sem{\autom}^F_k &\equiv \sem{\autom}_k \land \phi^k_F .
\end{align*}
Intuitively, $\sem{\autom}^F_k$ represents the \emph{accepted}
runs of $\autom$ of length 
exactly
$k+1$. 

Note that, if $k>0$, $\sem{\autom}^F_k$ is a formula over $\Sigma^{0\ldots
k-1}=\bigcup_{i=0}^{k-1}\Sigma^i$ and $\Gamma^{0\ldots
k}=\bigcup_{i=0}^k\Gamma^i$, to be interpreted over 
$\theory^{0\ldots k-1}=\bigcup_{i=0}^{k-1}\theory^i$. If $k=0$,
$\sem{\autom}^F_k$ is a formula over $\Gamma^0$ to be interpreted over the
empty theory.

We can use $\sem{\autom}^F_k$ to define an iterative procedure \ala Bounded
Model Checking~\cite{biere2009bounded}, shown in \cref{algo:semidecision}, where
at each $k\ge 0$ we test if $\sem{\autom}^F_k$ is satisfiable and we increment
$k$ if it is not. One can prove the following.
\begin{restatable}{proposition}{propunraveling}
\label{prop:unraveling}
$\lang_\theory(\autom) \neq \emptyset$ if and only if either $\sem{\autom}^F_0$
is satisfiable or $\sem{\autom}^F_k$ is $\theory^{0\ldots k-1}$-satisfiable for
some $k\ge 1$.
\end{restatable}

\begin{algorithm}[t]
  \caption{Non-emptiness semi-decision}
  \label{algo:semidecision}
  \begin{algorithmic}[1]
    \Procedure{NonEmpty}{$\autom$}
      \If{$\sem{\autom}^F_0$ is satisfiable}
        \State \Return \emph{not empty}
      \EndIf
      \For{$k\gets 1\ldots+\infty$}
        \If{$\sem{\autom}^F_k$ is $\theory^{0\ldots k-1}$-satisfiable}
          \State \Return \emph{not empty}
        \EndIf
      \EndFor
    \EndProcedure
  \end{algorithmic}
\end{algorithm}

Therefore, in order for the non-emptiness problem of first-order automata
to be semi-decidable, we need to test satisfiability of $\sem{\autom}^F_k$,
and we need this test to be decidable. This is, of course, not always
possible, but it can be depending on the shape of $\autom$ and the
considered theory.  Precisely, we give the following conditions for the
semi-decidability.

\begin{definition}[Conditions for semi-decidability of non-emptiness]
\label{thm:semidecision}
  The satisfiability of $\sem{\autom}^F_0$ is decidable and
  $\theory^{0\ldots k-1}$-satisfiability of $\sem{\autom}^F_k$ is decidable
  for all $k\ge 1$.
\end{definition}

It is simple to see that, if both conditions of~\cref{thm:semidecision} hold,
then \cref{algo:semidecision} is a semi-decision procedure for first-order
automata non-emptiness. This might seem a rather weak result in practice, but it
is worth noting that it is in fact a more general statement of what is already
implemented in the BLACK temporal reasoning
framework~\cite{GeattiGM19,GeattiGMV24,GeattiGM21} when dealing with \LTLf
\emph{modulo theories} (\LTLfMT), a restricted fragment of \FOLTL with rigid
predicates and non-rigid constants~\cite{GeattiGG22,GeattiGGW23}. By leveraging
modern SMT solvers~\cite{BarrettSST21}, BLACK manages to get promising
performance on satisfiable instances.

%% file: sections/8.decidable.tex

\section{Decidable Cases}
\label{sec:decidable}

First-order automata are a really general formalism and this generality is paid
with the undecidability of their emptiness. In this section, we show some cases
when decidability can be recovered by restricting what can appear in the
\emph{state signature} of the automata, and relate these results to some models
from the literature that can be captured by first-order automata at different
levels of restriction.

\begin{definition}[State signature restrictions]
  \label{def:control}
  An automaton $\autom$ with state signature $\Gamma$ is said to be:
  \begin{itemize}
    \item \emph{finite-control}, if $\Gamma$ consists of propositions;
    \item \emph{data-control}, if $\Gamma$ consists of propositions or
      constants;
    \item \emph{monadic}, if $\Gamma$ is a relational signature with only
      \emph{monadic} (\ie unary) predicates.
  \end{itemize}
\end{definition}

\noindent
\textbf{Finite-control automata.}
Finite-control automata are basically \emph{finite-state} machines, but can read
first-order alphabets. The finite number of different control states makes it
easier to handle the emptiness problem. 
\begin{restatable}{theorem}{thmdecisionfinitecontrol}
  \label{thm:decision:finite-control}
  Under the same conditions of \cref{thm:semidecision}, emptiness of
  finite-control automata is decidable.
\end{restatable}

Finite-control automata may seem limited but they are still quite expressive,
given the ability of treating infinite alphabets. Indeed, much literature is
devoted to \emph{symbolic finite automata}
(s-FA)~\cite{VeanesHLMB12,DAntoniV14,DAntoniKW18,DAntoniV17,TammV18}, which are
finite automata able to read large or infinite alphabets defined by effective
Boolean algebras. The state space of s-FAs is finite and emptiness is decidable
for them as well. Indeed, finite-control first-order automata can easily encode
s-FAs (see Appendix~\ref{app:encodings}). From a logical perspective, following
\cref{enc:encoding}, we can see that finite-control automata can encode \FOLTL
formulas where temporal operators are only applied to \emph{sentences} (\ie
closed formulas).

\smallskip

\noindent
\textbf{Data-control automata.}
In \emph{data-control} automata the state space returns to be infinite, but
loosely structured. Intuitively, the constants in $\Gamma$ act as storage
registers for single pieces of data coming from the respective domains. It turns
out this is sufficient to capture the \LTLfMT
logic~\cite{GeattiGG22,GeattiGGW23} mentioned in \cref{sec:emptiness}. This
logic restricts \FOLTL as follows: a) only constants can be non-rigid; b)
quantifiers cannot be applied to temporal formulas. At the same time, \LTLfMT
\emph{extends} \FOLTL with \emph{lookahead} term constructors $\nextvar c$ and
$\wnextvar c$, for some constant $c$, which evaluate to the value of $c$ at the
next time step (strongly or weakly, respectively, since we are on finite words).
For example, a predicate such as $p(\nextvar c)$ means that $p$ holds now on the
value that $c$ will have at the next state (which is required to exist). \FOLTL
can express lookaheads as follows:
\begin{equation*}
  p(\nextvar c) \equiv \exists x . (p(x) \land \ltl{X(c = x)})
\end{equation*}
By \cref{enc:encoding}, the automaton encoding this formula has a predicate in
$\Gamma$ acting as the surrogate for the formula $\ltl{X(c = x)}$. However, the
predicate always holds only for a single value, the one equal to $c$ at the next
state. Therefore, one can see that these predicates can be replaced by
constants, and thus \LTLfMT formulas can be captured by data-control automata.
This is interesting since it opens the way to transferring decidability results
from some known fragments of \LTLfMT~\cite{GeattiGGW23} to suitable subclasses
of data-control automata. Such fragments are based on a notion of \emph{finite
summary}, that intuitively requires the history of the word to be summarisable
at any time in a finite number of ways. A similar criterion is used in works
from the \emph{process modeling}
community~\cite{FelliMW22,GianolaMW23arxiv,GianolaMW24} to obtain decidability
of the emptiness problem for a class of state machines called \emph{data-aware
processes modulo theories} (DMT), whose control states are described by the
values of a finite set of variables over the theory's domain. Data-control
automata can capture DMTs quite easily, with the constants in $\Gamma$ doing the
job of such state variables (see Appendix~\ref{app:encodings} for details).
Therefore, we conjecture that a decidability criterion similar to DMT's finite
summary can be defined for data-control first-order automata as well.

It is also worth pointing out that data-control automata are suitable to
formalise the semantics of the MoXI modeling
language~\cite{JohannsenNDISTVR24,RDIJNSTV24}. MoXI is a general intermediate
language for the modeling of finite- and infinite-state model checking problems
that is recently being developed by the community in the open as a central
resource for research and experimentation. The language supports the description
of infinite-state systems through first-order models, but with the evolution
over time only of data variables. Although showing all the details of this
formalisation is outside of the scope of this paper, one can see that models
described through MoXI are covered by data-control automata.

\smallskip

\noindent
\textbf{Monadic automata.} From \cref{enc:encoding}, one can see that the whole
\FOLTL can be encoded using first-order automata whose state signature $\Gamma$
is purely relational. More precisely, the arity of the predicates in $\Gamma$
depends on the number of \emph{free variables} of temporal subformulas. If we
start from \emph{monodic} sentences, \ie those where temporal subformulas have
at most \emph{one} free variable, we obtain a state signature with only
\emph{monadic} predicates. This connection is interesting because of the
following classic result.
\begin{proposition}[Decidability of monodic \FOLTL~\protect\cite{GabbayKWZ03}]
  \label{prop:monodic}
  Let $\Sigma$ be a \emph{relational} signature with \emph{rigid
  constants}, $\fragment$ a class of first-order $\Sigma$-sentences
  \emph{without equality}, and $\theory$ a $\Sigma$-theory, such that
  $\theory$-satisfiability of $\fragment$-sentences is decidable. Then,
  $\theory$-satisfiability of \emph{monodic} \FOLTL $\fragment$-sentences
  is decidable.
\end{proposition} 

One may wonder whether \cref{prop:monodic} can be shown via monadic automata,
under the same assumptions. It turns out it can and, more precisely, we can show
that any monadic automaton can be translated into an equivalent finite-control
one, then using \cref{thm:decision:finite-control} to get decidability.

Behind this encoding is the observation that monadic predicates cannot express
any kind of relationship between different elements of the domains. The only
aspect one can model is the existence or not of elements satisfying certain
combinations of predicates, but elements which satisfy the exact same predicates
are indistinguishable. The domains can therefore be partitioned into a finite
number of classes, and this partition can be described by propositions. The
following result, with \cref{thm:decision:finite-control}, directly proves
\cref{prop:monodic}.

\begin{theorem}
  \label{thm:monodic:decidability}
  Let $\autom=\seq{\Sigma,\Gamma,\phi_0,\phi_T,\phi_F}$ be a \emph{monadic}
  automaton where $\phi_0$, $\phi_T$, and $\phi_F$ do not use equality.
  Then, there is a \emph{finite-control} automaton $\autom'$ equivalent to
  $\autom$.
\end{theorem}

%% file: sections/9.conclusions.tex

\section{Discussion, Related Work,\\and Future Directions}
\label{sec:conclusions}

We introduced \emph{first-order automata}, an automaton model capable of
capturing \FOLTL on finite words. We studied the corresponding notion of
$\theory$-regular languages, and how to deal with the (non-)emptiness problem in
some cases.

Of course, other models capable of describing infinite-state behaviors have been
proposed before. Many have been introduced in the last decades, including timed
automata~\cite{AlurD94}, hybrid automata~\cite{AlurCHH92}, recursive state
machines~\cite{AlurBEGRY05,BenerecettiMP10}, visibly pushdown
languages~\cite{AlurM04}, operator-precedence automata~\cite{DrosteDMP17}, FIFO
machines~\cite{BrandZ83}, counter machines~\cite{Wojna99}, Petri
nets~\cite{Murata89,HensenK09}, data-aware
processes~\cite{DeutschLV19,CalvaneseGM13,GhilardiGMR23,Gianola23}, many
automata over infinite alphabets~\cite{Segoufin06}, register
automata~\cite{KaminskiF94}, and many more. However, most of them have been
designed to address specific problems, in specific kinds of scenarios, with good
computational properties, so they cannot be general enough to capture \FOLTL
which is highly undecidable.

\emph{Symbolic finite automata}
(s-FA)~\cite{DAntoniKW18,DAntoniV14,DAntoniV17}, which we mentioned in
\cref{sec:decidable}, share some similarity with ours, although they are less
expressive because of their decidability. \emph{Alternating symbolic
automata}~\cite{IosifX19} are more expressive than s-FAs, closed under
complementation and with a semi-decidable emptiness problem, but are limited to
words that, in our setting, would be made of a signature $\Sigma$ with only
propositions and non-rigid constants.

Our automaton model is a valuable tool for reasoning about \FOLTL, as evidenced
by our proof of \cref{prop:monodic}. Although the original proof is more general
(it works also on $(\mathbb{Q},<)$, $(\mathbb{R},<)$, and any class of infinite
or finite first-order definable linear orders), ours is much more direct in the
particular case of finite discrete linear orders, which are those naturally
handled by automata. Automata are also crucial for \FOLTL model checking, that
can now be reduced to the emptiness problem. How to deal with this problem
efficiently in practice is a different issue. Building on top of existing SMT
technology~\cite{BarrettSST21}, a promising path would be to lift known
techniques, such as CEGAR~\cite{ClarkeGJLV00}, interpolation~\cite{McMillan18},
and property-directed reachability~\cite{Bradley12}. Constrained Horn
clauses~\cite{GurfinkelB19} are also promising for \emph{quantifier-free}
first-order automata, while for quantified ones, quantifier
elimination~\cite{CalvaneseGGMR22,CalvaneseGGMR21}, also of second-order
predicates~\cite{GabbaySS08}, has to be used.

Beyond verification, reactive synthesis~\cite{PnueliR89} is the most ambitious
goal. As we discussed extensively in \cref{sec:logic:purepast},
\emph{deterministic} automata are key for solving this problem. Most
promisingly, these can be obtained from pure-past \FOLTL sentences, as explained
in \cref{sec:logic:purepast}. Recent results \cite{RodriguezS23} addressed
reactive synthesis for simple fragments of \LTLfMT (hence of \FOLTL) with
encouraging results. For more general \FOLTL specifications, however, more work
is needed.

An extension of first-order automata to \emph{infinite words} is the natural
next step. One can easily define Büchi, Streett, or Rabin first-order
$\omega$-automata by replacing $\phi_F$ with suitable sentences representing the
sets of states used in such acceptance conditions. Then, \FOLTL on infinite
words is easily encoded into such Streett first-order $\omega$-automata by
extending \cref{enc:encoding}, although we leave this as future work. However,
first-order $\omega$-automata would pose even more challenges: all the theory
about $\omega$-regular languages has to be rebuilt from scratch, and dealing
with them algorithmically is going to be even harder.

%% file: sections/a.proofs.tex

\clearpage
\section{Proofs}
\label{app:proofs}

\cardinalitythm*
\begin{proof}
    Let $\letter$ be a model of $\theory$ and, for any $n>0$, let $\model_n = \seq{\letter_0,\dots,\letter_{n-1}}$ be such that $\letter_i = \letter$ for each $0\le i<n$. For every $M \subseteq \N \setminus \set{0}$, we define:
    \begin{equation}
        \lang_M \coloneqq \set{\model_n \suchthat n \in M}
    \end{equation}
    Clearly, for each $M \subseteq \N$, $\lang_M$ is a $\theory$-language, since $\lang_M \subseteq \sem{\theory}^*$. Moreover, there are $2^{\aleph_0}$ such distinct languages. Conversely, since a first-order automaton is a finite object, there are at most $\aleph_0$ many first-order automata over $\Sigma$. It follows that there are $\theory$-languages that are not $\theory$-regular.
\end{proof}

\unionintersection*
\begin{proof}
  Let $\theory$ be a $\Sigma$-theory and consider two automata
  $\autom_1=\seq{\Sigma,\Gamma_1,\phi^1_0,\phi^1_T,\phi^1_F}$ and
  $\autom_2=\seq{\Sigma,\Gamma_2,\phi^2_0,\phi^2_T,\phi^2_F}$, assuming
  \emph{w.l.o.g.}\ that $\Gamma_1$ and $\Gamma_2$ are \emph{disjoint}. 
  
  For the intersection, we can define an automaton
  $\autom_\cap=\seq{\Sigma,\Gamma_\cap,\phi_0^\cap,\phi_T^\cap,\phi_F^\cap}$
  such that
  $\lang_\theory(\autom_\cap)\!\!=\!\!\lang_\theory(\autom_1)\!\cap\!\lang_\theory(\autom_2)$:
  \begin{align*}
    \Gamma_\cap&=\Gamma_1\cup\Gamma_2 &
    \phi_0^\cap&=\phi^1_0\land\phi^2_0 \\
    \phi_T^\cap&=\phi^1_T\land\phi^2_T &
    \phi_F^\cap&=\phi^1_F\land\phi^2_F
  \end{align*}
  Now, let $\model=\seq{\letter_0,\ldots,\letter_n}$ be a word
  $\theory$-accepted by both $\autom_1$ and $\autom_2$, with the corresponding
  runs $\bar\rho_1=\seq{\rho^1_0,\ldots,\rho^1_n}$ and
  $\bar\rho'=\seq{\rho^2_0,\ldots,\rho^2_n}$. By construction, we have that
  $\rho^1_0\oplus\rho^2_0\models\phi^\cap_0$,
  $\rho^1_i\oplus\rho^2_i\oplus\letter_i\oplus\rho^{1\prime}_{i+1}\oplus\rho^{2\prime}_{i+1}\models\phi^\cap_T$,
  and $\rho^1_n\oplus\rho^2_n\models\phi_F^\cap$. Hence, the run
  $\bar\rho^\cap=\seq{\rho^1_0\oplus\rho^2_0,\ldots,\rho^1_n\oplus\rho^2_n}$ is
  a successful run for $\autom_\cap$, so $\autom_\cap$ $\theory$-accepts
  $\model$. Similarly, \viceversa, if $\model$ is $\theory$-accepted by
  $\autom_\cap$, it can be shown to be $\theory$-accepted by $\autom_1$ and
  $\autom_2$ as well. Hence,
  $\lang_\theory(\autom_\cap)=\lang_\theory(\autom_1)\cap\lang_\theory(\autom_2)$.

  For the union, we can define an automaton
  $\autom_\cup=\seq{\Sigma,\Gamma_\cup,\phi_0^\cup,\phi_T^\cup,\phi_F^\cup}$
  such that
  $\lang_\theory(\autom_\cup)=\lang_\theory(\autom_1)\cup\lang_\theory(\autom_2)$,
  as follows. Let $\Gamma_0=\set{p_0}$, where $p_0$ is a \emph{proposition}, and let
  $\Gamma_\cup=\Gamma_1\cup\Gamma_2\cup\Gamma_0$. Then, $\autom_\cup$ is defined
  as follows:
  \begin{align*}
    \phi_0^\cup &= (p_0\land\phi_0^1) \lor (\neg p_0 \land \phi_0^2)
    \\
    \phi_T^\cup &= (p_0\land \phi_T^1\land p_0') \lor (\neg p_0 \land \phi_T^2\land \neg p_0')
    \\
    \phi_F^\cup &= (p_0\land\phi_F^1) \lor (\neg p_0\land\phi_F^2)
  \end{align*}

  Now, let $\model=\seq{\letter_0,\ldots,\letter_n}$ be a word and assume
  \emph{w.l.o.g.}\ it is $\theory$-accepted by $\autom_1$, with the
  corresponding trace $\bar\rho_1=\seq{\rho^1_0,\ldots,\rho^1_{n+1}}$. Let
  $\bar\rho^\cup=\seq{\bar\rho^\cup_0,\ldots,\bar\rho^\cup_{n+1}}$, where each
  $\bar\rho^\cup_i$ is a structure over $\Gamma^\cup$ that agrees with
  $\rho^1_i$ on everything and sets $p_0$ true (to false if $\model$ was
  $\theory$-accepted by $\autom_2$ instead). Then, by construction, we have that
  $\bar\rho^\cup_0\models\phi_0^\cup$,
  $\rho^\cup_i\oplus\letter\oplus\rho^{\cup'}_{i+1}\models\phi_T^\cup$, and
  $\rho^\cup_{n+1}\models\phi_F^\cup$. Hence, $\bar\rho^\cup$ is a successful
  run for $\autom_\cup$. Similarly, \viceversa, one can see that if $\model$ is
  $\theory$-accepted by $\autom_\cup$, then it is $\theory$-accepted by either
  $\autom_1$ or $\autom_2$ (depending on the truth of $p_0$ in the corresponding
  run). Hence,
  $\lang_\theory(\autom_\cup)=\lang_\theory(\autom_1)\cup\lang_\theory(\autom_2)$.
\end{proof}

\thmsecondorderclass*
\begin{proof}
  Let $\phi\equiv\exists P_1\ldots\exists P_n\suchdot \psi$ be the existential
  second-order $\Sigma$-sentence defining $\mathcal{C}$, where $\psi$ is a
  first-order sentence over $\Sigma\cup\set{P_1,\ldots,P_n}$. Then, we define
  the first-order automaton
  $\autom_\phi=\seq{\Sigma,\Gamma,\phi_0,\phi_T,\phi_F}$ such that
  $\Gamma=\set{P_1,\ldots,P_n}$, $\phi_0\equiv\phi_F\equiv\top$, and
  $\phi_T\equiv \psi$. It can be confirmed that
  $\lang_\theory(\autom)=L_{\mathcal{C}}$.
\end{proof}

\secondorderautomata*
\begin{proof}
  Let $\autom=\seq{\Sigma,\Gamma,\phi_0,\phi_T,\phi_F}$ be a
  $\Sigma_1^1$-automaton, and let $\phi_0\equiv Q_0\suchdot \psi_0$,
  $\phi_T\equiv Q_T\suchdot \psi_T$ and $\phi_F\equiv Q_F\suchdot \psi_F$, where
  $Q_0$, $Q_T$, and $Q_F$ are blocks of existential second-order quantifiers,
  and $\psi_0$, $\psi_T$ and $\psi_F$ are first-order sentences over a signature
  containing the symbols mentioned by $Q_0$, $Q_T$ and $Q_F$, respectively. We
  define $\autom_1=\seq{\Sigma,\Gamma_1,\psi_0,\psi_T,\psi_F}$, where $\Gamma_1$
  is built from $\Gamma$ by adding the second-order symbols quantified in $Q_0$,
  $Q_T$ and $Q_F$. Now, for any $\Sigma$-theory $\theory$ we can prove that
  $\lang_\theory(\autom_1)=\lang_\theory(\autom)$.

  $(\subseteq)$ Suppose we have a $\theory$-word
  $\bar\sigma=\seq{\sigma_0,\ldots,\sigma_{n-1}}$ in $\lang_\theory(\autom_1)$.
  Then there is a successful run $\bar\rho^1=\seq{\rho^1_0,\ldots,\rho^1_n}$,
  where each $\rho^1_i$ is a $\Gamma_1$-structure. Let
  $\bar\rho=\seq{\rho_0,\ldots,\rho_n}$ be a sequence of $\Gamma$-structures
  where each $\rho_i$ is the $\Gamma$-substructure of $\rho_i^1$. Then, since
  $\rho^1_0\models \psi_0$, almost by definition it holds that $\rho_0\models
  Q_0\suchdot\psi_0$, \ie $\rho_0\models \phi_0$. Similarly, since
  $\rho^1_i\oplus\sigma_i\oplus\rho_{i+1}^{1\prime}\models\psi_T$, it also holds
  that $\rho_i\oplus\sigma_i\oplus\rho'_{i+1}\models\phi_T$, and similarly we
  have that $\rho_n\models\phi_F$. Hence $\bar\rho$ is a successful run for
  $\autom$, which then accepts $\bar\sigma$.

  $(\supseteq)$ Suppose we have a $\theory$-word
  $\bar\sigma=\seq{\sigma_0,\ldots,\sigma_{n-1}}$ in $\lang_\theory(\autom)$. So
  there is a successful run $\bar\rho=\seq{\rho_0,\ldots,\rho_n}$, where each
  $\rho_i$ is a $\Gamma$-structure. Since $\rho_0\models\phi_0$, there is an
  $\Gamma_1$-expansion $\rho^1_0$ of $\rho_0$ such that $\rho^1_0\models\psi_0$.
  Similarly, for each $0\le i<n$, since
  $\rho_i\oplus\sigma_i\oplus\rho_{i+1}'\models\phi_T$, there is a
  $\Gamma_1$-expansion $\rho^1_i$ of $\rho_i$ and a $\Gamma_1'$-expansion
  $\rho^1_{i+1}$ of $\rho_{i+1}$ such that
  $\rho^1_i\oplus\sigma_i\oplus\rho_{i+1}^{1\prime}\models\psi_T$. Note that
  since $\phi_T$ does not contain the primed symbols in
  $\Gamma_1'\setminus\Gamma'$, all the $\Gamma_1'$-expansions of $\rho_{i+1}'$
  are suitable here. Now, we have that since $\rho_n\models\phi_F$ we have a
  $\Gamma_1$-expansion $\rho^1_n$ of $\rho_n$ such that $\rho^1_n\models\psi_F$.
  Hence there exist this sequence $\bar\rho^1=\seq{\rho_0^1,\ldots,\rho^1_n}$ of
  $\Gamma_1$-expansions of $\bar\rho$, which is a successful run in $\autom_1$,
  which then accepts $\bar\sigma$.
\end{proof}

\thmconcatenation*
\begin{proof}
  Consider two first-order automata, respectively
  $\autom_1=\seq{\Sigma,\Gamma_1,\phi^1_0,\phi^1_T,\phi^1_F}$ and
  $\autom_2=\seq{\Sigma,\Gamma_2,\phi^2_0,\phi^2_T,\phi^2_F}$. We can build a
  FO-automaton $\autom=\seq{\Sigma,\Gamma,\phi_0,\phi_T,\phi_F}$ such
  that
  $\lang_\theory(\autom)=\lang_\theory(\autom_1)\cdot\lang_\theory(\autom_2)$
  for any $\Sigma$-theory $\theory$:
  \begin{enumerate}
    \item $\Gamma = \Gamma_1 \cup \Gamma_2 \cup \set{ p }$ where $p$ is a fresh
      proposition
    \item $\phi_0\equiv (p \land \phi^1_0) \lor (\neg p \land \phi_0^1 \land \phi_F^1 \land \phi_0^2)$
    \item $\phi_F\equiv \neg p \land \phi^2_F$
    \item the transition relation $\phi_T$ is defined as follows:
      \begin{align*}
        \phi_T \equiv {} & (p \land \phi^1_T \land p') \lor {} \\
        & (\neg p \land \phi^2_T \land \neg p') \lor {} \\ 
        &
          \Bigl(
            p \land \neg p' \land
            \phi_T^1 \land \phi_F^1[\Gamma_1/\Gamma'_1] 
            \land \phi_0^2[\Gamma_2/\Gamma'_2]
          \Bigr)
      \end{align*}
  \end{enumerate}
  The proposition $p$ indicates which of the two automata is currently being simulated. The first two disjuncts of $\phi_T$ simulate transitions of $\autom_1$ and $\autom_2$, respectively. The third disjunct performs the switch from $\autom_1$ to $\autom_2$: the current letter is processed by $\autom_1$, the resulting $\Gamma_1$-component is required to satisfy $\phi_F^1$, and the $\Gamma_2$-component of the destination state is required to satisfy $\phi_0^2$. Hence, the switch occurs immediately after the last letter of the factor accepted by $\autom_1$.

  The second disjunct of $\phi_0$ accounts for the case $\epsilon\in \lang(\autom_1)$, allowing the run to start directly in the $\autom_2$-phase. In particular, if also $\epsilon\in \lang(\autom_2)$, the initial state can already satisfy $\phi_F$, so that $\autom$ accepts $\epsilon$.

  It follows directly from the construction that every word $uv$, where $u\in \lang(\autom_1)$, $v\in \lang(\autom_2)$, and $uv$ is a $\theory$-word, induces an accepting run of $\autom$. Conversely, in every accepting run of $\autom$, $p$ either is false from the initial state, in which case $\epsilon\in \lang(\autom_1)$, or changes exactly once from true to false through the third disjunct of $\phi_T$. In the latter case, the prefix up to and including the letter on which the switch occurs is accepted by $\autom_1$, while the remaining suffix is accepted by $\autom_2$. Therefore, $\lang(\autom) = \lang(\autom_1)\cdot\lang(\autom_2)$. Hence $\theory$-regular languages are closed under concatenation.
\end{proof}

\thmkleene*
\begin{proof}
  Let $\autom=\seq{\Sigma,\Gamma,\phi_0,\phi_T,\phi_F}$ be a first-order
  automaton. We define a $\Sigma^1_1$-automaton
  $\autom^*=\seq{\Sigma,\Gamma^*,\phi_0^*,\phi_T^*,\phi_F^*}$ as follows.
  Let $\Gamma^*=\Gamma\cup\set{p}$, where $p$ is a fresh proposition, and
  define $\phi_0^* \equiv \phi_F^* \equiv p$.  The proposition $p$
  identifies a fresh control state representing the point between two
  consecutive factors. Since this state is both initial and final,
  $\autom^*$ accepts the empty word.

  Let $\alpha \coloneqq \neg p \lor (p\land\phi_0)$.  The transition
  relation is defined as follows:
  \begin{align*}
    \phi_T^* \equiv {}&
      (
        \alpha \land \phi_T \land \neg p'
      ) \lor {}\\
      & \Bigl(
        \alpha \land p' \land \exists\Gamma''
          (\phi_T[\Gamma'/\Gamma''] \land \phi_F[\Gamma/\Gamma''])
        \Bigr)
  \end{align*}

  The first disjunct simulates a transition of $\autom$ without completing
  the current factor. If $p$ holds, $\phi_0$ additionally requires the
  current $\Gamma$-state to be initial, thereby starting a new run of
  $\autom$.  The second disjunct completes the current factor. It requires
  the existence of a $\Gamma''$-state reachable through $\phi_T$ that
  satisfies $\phi_F$, and then sets $p'$ to true. Thus, after reading the
  last letter of a factor belonging to $\lang(\autom)$, the run returns to
  the distinguished state from which another factor can be started. Notice
  that this disjunct can also be taken when $p$ holds, thereby allowing
  a word of length one to constitute a complete factor.

  It follows that every concatenation $u_1\cdots u_k$ where $u_i\in
  \lang(\autom)$ and $u_1\cdots u_k$ is a $\theory$-word, induces an
  accepting run of $\autom^*$: the run starts with $p$, simulates an
  accepting run of $\autom$ on each $u_i$, and returns to $p$ after the
  last letter of every factor. The case $k=0$ corresponds to the empty
  word, which is accepted since $p$ is both the initial and final
  condition.

  Conversely, consider a non-empty word accepted by $\autom^*$. Since the
  initial and final states satisfy $p$, every accepting run can be
  partitioned at the transitions whose destination satisfies $p$. Between
  two consecutive such positions, the transition relation simulates a run
  of $\autom$ starting from a state satisfying $\phi_0$, while the
  transition returning to $p$ witnesses a successor satisfying $\phi_F$.
  Hence each corresponding factor belongs to $\lang(\autom)$. Therefore,
  $\lang(\autom^*) = \lang(\autom)^*$.
  By~\cref{thm:second-order:automata}, $\autom^*$ is equivalent to
  a first-order automaton. Hence $\theory$-regular languages are closed
  under Kleene star.
\end{proof}

\detuniquerun*
\begin{proof}
  Let $\bar\sigma=\seq{\sigma_0,\ldots,\sigma_{n-1}}$ be a nonempty $\theory$-word.
  Given the interpretation that $\bar\sigma$ gives to the sorts in the
  signature, by $\autom$ being both complete and deterministic we know
  there is a unique model $\rho_0$ interpreting the sorts in the same way and
  satisfying $\phi_0$. Then, determinism imposes that $\Gamma'$ is implicitly
  defined by $\Gamma\cup\Sigma$ under $\phi_T$, \ie for any
  $(\Gamma\cup\Sigma)$-structure $\mu$ there is at most one
  $(\Gamma\cup\Sigma\cup\Gamma')$-expansion $\tau$ such that
  $\tau\models\phi_T$. For completeness, there is at least one such $\tau$. So
  $\tau$ exists and is unique. So recursively we define
  $\bar\rho=\seq{\rho_0,\ldots,\rho_n}$ as follows. For $0\le i < n$, let
  $\mu_i=\rho_i\oplus\sigma_i$, consider the unique
  $(\Gamma\cup\Sigma\cup\Gamma')$-expansion $\tau_i$ of $\mu_i$, and define
  $\rho_{i+1}$ as the restriction of $\tau_i$ over $\Gamma'$ renamed to
  $\Gamma$. By construction $\bar\rho$ is the unique run of $\autom$ induced by
  $\bar\sigma$.

  Finally, the fact that all runs induced by the empty word agree on
  acceptance derives from Item~3 of~\cref{def:determinism}.
\end{proof}

\detcomplementation*
\begin{proof}
  Let $\autom=\seq{\Sigma,\Gamma,\phi_0,\phi_T,\phi_F}$ be a
  \emph{deterministic} automaton. As in the finite-state case, a complete and
  deterministic automaton is very easy to complement: just swap accepting and
  non-accepting states. To see that it works, let
  $\overline{\autom}=\seq{\Sigma,\Gamma,\phi_0,\phi_T,\neg\phi_F}$. We show that a
  $\theory$-word is accepted by $\overline{\autom}$ if and only if it is \emph{not}
  accepted by $\autom$. 

  Let $\bar\sigma$ be a nonempty $\theory$-word.  On the one hand, if
  $\bar\sigma$ is accepted by $\overline{\autom}$, there is a run
  $\bar\rho=\seq{\rho_0,\ldots,\rho_n}$ induced by $\bar\sigma$, which is
  unique by \cref{prop:determinism:unique:run}, such that, in particular
  $\rho_n\models\neg\phi_F$. But then $\bar\rho$ is not an accepting run
  for $\autom$, and since $\bar\rho$ is unique, there are no other induced
  runs that can make $\autom$ accept $\bar\sigma$. The \viceversa is
  similar. 

  It remains to consider the empty word. By Item~3
  of~\cref{def:determinism}, either $\phi_0\to\phi_F$ or
  $\phi_0\to\neg\phi_F$.  In the former case, every initial state is
  accepting for $\autom$ and none is accepting for $\overline{\autom}$.
  Hence $\epsilon\in\lang_\theory(\autom)$ and
  $\epsilon\notin\lang_\theory(\overline{\autom})$. In the latter case, no
  initial state is accepting for $\autom$, while every initial state is
  accepting for $\overline{\autom}$. Hence
  $\epsilon\notin\lang_\theory(\autom)$ and
  $\epsilon\in\lang_\theory(\overline{\autom})$. Finally, completeness
  guarantees that initial states exist.

  Therefore, we have
  $\lang_\theory(\overline{\autom})=\overline{\lang_\theory(\autom)}$.
\end{proof}

\encodingproof*
\begin{proof}
  Let $\phi$ be a \FOLTL $\Sigma$-sentence, $\theory$ be a $\Sigma$-theory,
  and $\autom(\phi)=\seq{\Sigma,\Gamma,\phi_0,\phi_T,\phi_F}$ be the
  first-order automaton built according to~\cref{enc:encoding}.

  We start by noticing that both the languages $\lang_\theory(\phi)$ and
  $\lang_\theory(\autom(\phi))$ exclude the empty word: \FOLTL sentences
  are interpreted over nonempty words, while the condition
  $\phi_0\land\phi_F$ is inconsistent, since it requires both $\xs_\phi$
  and $\lnot\xs_\phi$.

  Let $\bar\sigma=\seq{\letter_0,\ldots,\letter_{n-1}}$ be a nonempty
  $\theory$-word. We say that $\bar\rho = \seq{\rho_0,\dots,\rho_n}$ is a
  \emph{final run} of $\autom(\phi)$ induced by $\bar\sigma$ iff $\bar\rho$ is
  an accepting run of $\autom(\phi)$ induced by $\bar\sigma$ except for the fact
  that we do \emph{not} require $\rho_0$ to satisfy $\xs_\phi$. We first prove
  the following claim, establishing a one-to-one correspondence between
  satisfaction of formulas in a $\theory$-word and their satisfaction in the
  final runs induced by it in $\autom(\phi)$.
  
  \begin{claim}
  \label{claim:encoding:finalrun}
    For each nonempty $\theory$-word
    $\bar\sigma=\seq{\letter_0,\ldots,\letter_{n-1}}$ and each final run
    $\bar\rho = \seq{\rho_0,\dots,\rho_n}$ of $\autom(\phi)$ induced by
    $\bar\sigma$, it holds that:
    \begin{equation*}
      \rho_i\oplus\letter_i\oplus\rho'_{i+1},\xi \models \snf_S(\psi)
      \quad\Leftrightarrow\quad
      \bar\sigma,\xi,i \models \psi
    \end{equation*}
    for every $\psi\in \closure(\phi)$, every $0\leq i<n$, and every
    environment $\xi$.
  \end{claim}
  \begin{claimproof}
    We prove the claim by induction on the structure of $\psi$.  For the
    base case, let $\psi = p(\bar x)$ be an atomic formula. Since
    $\snf_S(p(\bar x)) = p(\bar x)$ and $p \in \Sigma$, the
    $\Sigma$-structure $\letter_i$ and the environment $\xi$ suffice for
    the formula $\snf_S(p(\bar x))$ to be interpreted by
    $\rho_i\oplus\letter_i\oplus\rho'_{i+1},\xi$.
    Therefore, $\rho_i\oplus\letter_i\oplus\rho'_{i+1},\xi \models p(\bar
    x)$ iff $\letter_i,\xi \models p(\bar x)$ iff $\bar\sigma,\xi,i \models
    p(\bar x)$. For the inductive case, we distinguish the kind of formula
    at hand.
    \begin{enumerate}
      \item Take the formula $\forall x \suchdot \psi(\bar x)$ and take
        $\rho_i\oplus\letter_i\oplus\rho'_{i+1},\xi \models \snf_S(\forall
        x \suchdot \psi(\bar x))$.  By definition of $\snf_S$, this is
        equivalent to $\rho_i\oplus\letter_i\oplus\rho'_{i+1},\xi \models
        \forall x \suchdot \snf_S(\psi(\bar x))$, that is
        $\rho_i\oplus\letter_i\oplus\rho'_{i+1},\xi' \models
        \snf_S(\psi(\bar x))$, for all environments $\xi'$ that agree with
        $\xi$ except possibly on $x$. Then, by induction hypothesis, we
        have that this is equivalent to $\bar\sigma,\xi',i \models \psi$
        and therefore $\bar\sigma,\xi,i \models \forall x \suchdot
        \psi(\bar x)$.
      \item The cases of Boolean connectives and existential quantifiers
        follow the same reasoning as for the previous point.
      \item Take the formula $\ltl{X \psi(\bar x)}$ and consider
        $\rho_i\oplus\letter_i\oplus\rho'_{i+1},\xi\models
        \snf_S(\ltl{X\psi(\bar x)})$. Since $\snf_S(\ltl{X\psi(\bar x)})
        = \xs'_{\psi}(\bar x)$, this is equivalent to $\rho'_{i+1},\xi
        \models \xs'_{\psi}(\bar x)$. Since $\bar\rho$ is a final run, we
        know that $\rho_n$ satisfies $\phi_F$, which requires $\forall
        x.\neg \xs_\psi(\bar x)$, therefore we know that $i<n-1$. Now,
        $\rho'_{i+1},\xi\models \xs'_{\psi}(\bar x)$ holds if and only if
        $\rho_{i+1},\xi\models \xs_{\psi}(\bar x)$, hence by expanding
        $\rho_{i+1}$ to $\rho_{i+1}\oplus\letter_{i+1}\oplus\rho'_{i+2}$ we
        trivially also get that
        $\rho_{i+1}\oplus\letter_{i+1}\oplus\rho'_{i+2},\xi\models
        \xs_{\psi}(\bar x)$, and by $\phi_T$ it follows that this is
        equivalent to
        $\rho_{i+1}\oplus\letter_{i+1}\oplus\rho'_{i+2},\xi\models
        \snf_S(\psi(\bar x))$. By the induction hypothesis we have
        $\model,\xi,i+1\models \psi(\bar x)$, which means
        $\model,\xi,i\models\ltl{X\psi(\bar x)}$. 
      \item Take the formula $\ltl{Y \psi(\bar x)}$ and consider
        $\rho_i\oplus\letter_i\oplus\rho'_{i+1},\xi\models
        \snf_S(\ltl{Y\psi(\bar x)})$. Since $\snf_S(\ltl{Y\psi(\bar x)})
        = \ys_{\psi}(\bar x)$, this is equivalent to $\rho_i,\xi \models
        \ys_{\psi}(\bar x)$. Since $\bar\rho$ is a final run, we know that
        $\rho_0$ satisfies $\forall x.\neg \ys_\psi(\bar x)$, therefore we
        know that $i>0$. Now, $\rho_{i},\xi\models \ys_{\psi}(\bar x)$
        holds if and only if $\rho'_{i},\xi\models \ys'_{\psi}(\bar x)$,
        hence by expanding $\rho'_{i}$ to
        $\rho_{i-1}\oplus\letter_{i-1}\oplus\rho'_{i}$ we trivially also
        get that $\rho_{i-1}\oplus\letter_{i-1}\oplus\rho'_{i},\xi\models
        \ys'_{\psi}(\bar x)$, and by $\phi_T$ it follows that this is
        equivalent to
        $\rho_{i-1}\oplus\letter_{i-1}\oplus\rho'_{i},\xi\models
        \snf_S(\psi(\bar x))$. By the induction hypothesis we have
        $\model,\xi,i-1\models \psi(\bar x)$, which means
        $\model,\xi,i\models\ltl{Y\psi(\bar x)}$. 
      \item The cases for $\ltl{wX\psi}$ and $\ltl{Z\psi}$ follow the same
        reasoning as for the previous two points.
      \item Take the formula $\ltl{\psi_1(\bar x) U \psi_2(\bar y)}$ and
        consider $\rho_i\oplus\letter_i\oplus\rho'_{i+1},\xi\models
        \snf_S(\ltl{\psi_1(\bar x) U \psi_2(\bar y)})$. Recall that
        $\snf_S(\ltl{\psi_1(\bar x) U \psi_2(\bar y)}) = \snf_S(\psi_2(\bar
        y))\lor(\snf_S(\psi_1(\bar x))\land\xs'_{\ltl{\psi_1
        U \psi_2}}(\bar x, \bar y))$. Now, either we have
        $\rho_i\oplus\letter_i\oplus\rho'_{i+1},\xi\models\snf_S(\psi_2(\bar
        y))$, or $\rho_i\oplus\letter_i\oplus\rho'_{i+1},\xi\models
        \snf_S(\psi_1(\bar x))\land \xs'_{\ltl{\psi_1 U \psi_2}}(\bar x,
        \bar y)$.  In the latter case, by $\phi_T$ we have
        $\rho_{i+1}\oplus\letter_{i+1}\oplus\rho'_{i+2},\xi\models\snf_S(\ltl{\psi_1(\bar
        x)U\psi_2(\bar y)})$ again, and this can continue for all the word
        up to $i=n-1$, where, however, we must have
        $\rho_{n-1}\oplus\letter_{n-1}\oplus\rho'_n,\xi\models\snf_S(\psi_2(\bar
        y))$, because $\phi_F$ requires $\xs_{\ltl{\psi_1 U \psi_2}}(\bar
        x, \bar y)$ to be false over $\rho_n$. So there is a $j$ with $i\le
        j<n$ such that $\rho_j\oplus\letter_j\oplus\rho'_{j+1},\xi \models
        \snf_S(\psi_2(\bar y))$, and
        $\rho_k\oplus\letter_k\oplus\rho'_{k+1},\xi\models\snf_S(\psi_1(\bar
        x))$ for all $i\le k < j$. By the induction hypothesis applied to
        both facts we get that $\model,\xi,j\models \psi_2(\bar y)$ and
        $\model,\xi,k\models \psi_1(\bar x)$. Hence $\model,\xi,i\models
        \ltl{\psi_1(\bar x) U \psi_2(\bar y)}$.
      \item The cases for the $\ltl{R}$, $\ltl{S}$, and $\ltl{T}$ operators
        follow the same reasoning as for the previous point.
    \end{enumerate}
  \end{claimproof}

  In the next claim, we prove that, \emph{for all nonempty $\theory$-words}
  $\bar\sigma \in \sem{\theory}^+$, there exists a final run of
  $\autom(\phi)$ induced by $\bar\sigma$, no matter whether $\bar\sigma$ is
  a model of $\phi$.

  \begin{claim}
  \label{claim:encoding:finalrun:existence}
    For each nonempty $\theory$-word
    $\bar\sigma=\seq{\letter_0,\ldots,\letter_{n-1}}$, there exists
    a final run $\widehat\rho$ of $\autom(\phi)$ induced by $\bar\sigma$.
  \end{claim}
  \begin{claimproof}
    We construct a sequence $\widehat{\rho}=\seq{\widehat{\rho}_0,\ldots,
    \widehat{\rho}_n}$ for the $\theory$-word $\bar\sigma$, and prove that
    $\widehat{\rho}$ is a final run. Clearly, first we require that the sort
    symbols in each $\widehat{\rho}_i$ are interpreted as in $\letter_i$.
    Second, for every future surrogate $s_\psi \in \XS\cup\wXS$ and every
    $0\leq i<n$, we define its interpretation by requiring, for every
    environment $\xi$,
    \begin{align*}
      \widehat{\rho}_i,\xi \models s_\psi(\bar x)
        \quad\Leftrightarrow\quad
      \bar\sigma,\xi,i\models\psi(\bar x).
    \end{align*}
    At $\widehat{\rho}_n$, we interpret each $\xs_\psi$ as the empty
    relation and each $\ws_\psi$ as the corresponding full relation (over
    the appropriate sorts interpretations).
    Third, for every past surrogate $s_\psi\in\YS \cup \ZS$ and every
    $0\le i< n$, require, for every environment $\xi$,
    \begin{align*}
      \widehat{\rho}_{i+1},\xi\models s_\psi(\bar x)
        \quad\Leftrightarrow\quad
      \bar\sigma,\xi,i\models\psi(\bar x).
    \end{align*}
    At $\widehat{\rho}_0$, we interpret each $\ys_\psi$ as the empty
    relation and each $\zs_\psi$ as the corresponding full relation (over
    the appropriate sorts interpretations).

    Now, we verify that $\widehat{\rho}$ is indeed a final run of
    $\autom(\phi)$ induced by $\bar\sigma$, that is: 
    \begin{enumerate*}[label=(\roman*)]
      \item $\widehat{\rho}_0 \models \phi'_0$, with $\phi'_0$ defined as
        the modification of $\phi_0$ in which we remove the conjunct
        $\xs_\phi$;
      \item $\widehat{\rho}_n \models \phi_F$; and
      \item $\widehat{\rho}_i \oplus \letter_i \oplus \widehat{\rho}'_{i+1}
        \models \phi_T$, for each $0\le i < n$.
    \end{enumerate*}
    The first two points are straightforward to verify: by construction of
    $\widehat{\rho}$, each future surrogate $\xs_\psi(\bar x) \in \XS$
    (resp., $\ws_\psi(\bar x) \in \wXS$) is false (resp., is true) in
    $\widehat{\rho}_n$, and each past surrogate $\ys_\psi(\bar x) \in \YS$
    (resp., $\zs_\psi(\bar x) \in \ZS$) is false (resp., is true) in
    $\widehat{\rho}_0$.

    It remains to verify the transition condition $\phi_T$. An easy
    induction on the structure of the formula shows that, for every $\psi
    \in \closure(\phi)$, every $0\le i<n$, and every environment $\xi$,
    \begin{equation}
    \label{eq:constructed-sequence-snf}
      \widehat{\rho}_i \oplus \sigma_i \oplus \widehat{\rho}'_{i+1},\xi
      \models \snf_S(\psi)
        \quad\Leftrightarrow\quad
      \bar\sigma,\xi,i \models \psi.
    \end{equation}
    For atomic formulas, including equalities, the equivalence follows
    because $\widehat{\rho}_i\oplus \sigma_i \oplus \widehat{\rho}'_{i+1}$
    and $\sigma_i$ agree on $\Sigma$. 
    If $\psi = \ltl{X \psi_1}$, then $\bar\sigma,\xi,i \models \psi$ iff
    $i<n$ and $\bar\sigma,\xi,i+1 \models \psi_1$. If $i=n-1$, both sides
    of~\eqref{eq:constructed-sequence-snf} are false: indeed,
    $\snf_S(\psi)=\xs'_{\psi_1}$, and $\xs_{\psi_1}$ is interpreted as the
    empty relation in $\widehat{\rho}_n$.  Otherwise, by definition of the surrogate $\xs_{\psi_1}$, $\bar\sigma,\xi,i+1 \models \psi_1$ is equivalent to
    $\widehat{\rho}_{i+1} \oplus \sigma_{i+1} \oplus \widehat{\rho}'_{i+2},\xi
    \models \xs_{\psi_1}$, that is, $\widehat{\rho}_{i+1},\xi \models
    \xs_{\psi_1}$, which is equivalent to $\widehat{\rho}'_{i+1},\xi
    \models \xs'_{\psi_1}$ and, since $\xs'_{\psi_1}= \snf_S(\ltl{X \psi_1})= \snf_S(\psi) $ also to
    $\widehat{\rho}_i \oplus \sigma_i \oplus \widehat{\rho}'_{i+1},\xi
    \models \snf_S(\psi)$. The proof for the cases of the other types of
    formulas uses the induction hypotheses, but is completely analogous.

    We can now verify every conjunct of $\phi_T$.  For each future
    surrogate $s_\psi \in \XS\cup\wXS$, its
    interpretation in $\widehat{\rho}_i$ and
    \eqref{eq:constructed-sequence-snf} give
    \begin{align*}
      &\widehat{\rho}_i \oplus \letter_i \oplus \widehat{\rho}'_{i+1},\xi \models s_\psi(\bar x)
        \Leftrightarrow \\
      &\bar\sigma,\xi,i \models \psi(\bar x)
        \Leftrightarrow \\
      &\widehat{\rho}_i \oplus \letter_i \oplus \widehat{\rho}'_{i+1},\xi \models
      \snf_{S}(\psi(\bar x)).
    \end{align*}
    Similarly, for each past surrogate $s_\psi \in \YS\cup\ZS$, its
    interpretation in $\widehat{\rho}_{i+1}$ gives
    \begin{align*}
      &\widehat{\rho}_i \oplus \letter_i \oplus \widehat{\rho}'_{i+1},\xi \models s'_\psi(\bar x)
        \Leftrightarrow \\
      &\bar\sigma,\xi,i \models \psi(\bar x)
        \Leftrightarrow \\
      &\widehat{\rho}_i \oplus \letter_i \oplus \widehat{\rho}'_{i+1},\xi \models \snf_{S}(\psi(\bar x)).
    \end{align*}
    Since these equivalences hold for every environment $\xi$, all the
    equivalences in $\phi_T$ are satisfied.
    Therefore,
    \begin{align*}
      \widehat{\rho}_i \oplus \sigma_i
      \oplus \widehat{\rho}'_{i+1}
      \models \phi_T
      \qquad\text{for every }0 \leq i < n.
    \end{align*}
    This proves that $\widehat{\rho}$ is a final run of $\autom(\phi)$
    induced by $\bar\sigma$.
  \end{claimproof}

  To conclude the proof, let
  $\bar\sigma=\seq{\sigma_0,\ldots,\sigma_{n-1}}$ be a nonempty
  $\theory$-word.  For every final run $\rho=\seq{\rho_0,\ldots,\rho_n}$ of
  $\autom(\phi)$ induced by $\bar\sigma$, we have that: 
  \begin{align*}
    \rho_0 \models \xs_\phi
    &\Leftrightarrow
    \rho_0 \oplus \sigma_0 \oplus \rho'_1
      \models \snf_S(\phi)\\
    &\Leftrightarrow
    \bar\sigma,0 \models \phi,
  \end{align*}

  \noindent where in the first `iff' we used the transition conjunct
  associated with $\xs_\phi$ (remember that by definition of $\phi_T$, 
  $\xs_\phi$ is equivalent to $\snf_S(\phi)$), whereas in the second `iff' we used~\cref{claim:encoding:finalrun}.

  By definition, a final run already satisfies all the requirements for
  acceptance except possibly the initial conjunct $\xs_\phi$. Therefore,
  a final run induced by $\bar\sigma$ is accepting if and only if
  $\bar\sigma,0 \models \phi$.

  If $\bar\sigma \in \lang_\theory(\phi)$,
  \cref{claim:encoding:finalrun:existence} provides a final run, which is
  accepting by the above equivalence. Hence $\bar\sigma \in
  \lang_\theory(\autom(\phi))$.  Conversely, if $\bar\sigma \in
  \lang_\theory(\autom(\phi))$, any accepting run is in particular a final
  run and satisfies $\rho_0 \models \xs_\phi$.  The same equivalence
  therefore yields $ \bar\sigma,0 \models \phi$, which means $\bar\sigma \in \lang_\theory(\phi)$.
  Since the empty word belongs to neither language, we conclude
  $\lang_\theory(\autom(\phi)) = \lang_\theory(\phi)$. 
\end{proof}

\encodingpastproof*
\begin{proof}
  Let $\autom(\phi)= \seq{\Sigma,\Gamma,\phi_0,\phi_T,\phi_F}$ be the
  first-order automaton obtained by~\cref{enc:past:encoding}, and let
  $\bar\sigma=\seq{\sigma_0,\ldots,\sigma_{n-1}}$ be a nonempty
  $\theory$-word.

  Since $\phi$ is pure-past, every formula $\snf_S(\psi)$ occurring in
  $\phi_T$ uses only symbols from $\Gamma\cup\Sigma$, and no primed state
  predicates. With the domains of the sorts fixed by $\bar\sigma$, the
  initial condition uniquely determines $\rho_0$, interpreting every
  $\ys_\psi$ as the empty relation and every $\zs_\psi$ as the full
  relation. Moreover, each transition equivalence
  \begin{align*}
    \forall\bar x\,
    \bigl(s'_\psi(\bar x)\leftrightarrow
          \snf_S(\psi(\bar x))\bigr)
  \end{align*}
  uniquely determines the corresponding predicate in the next state from
  the current state and input letter. Consequently, $\bar\sigma$ induces
  a unique run $\rho=\seq{\rho_0,\ldots,\rho_n}$.

  We prove by induction on $i$, for $0\leq i\leq n$, that the following
  invariant holds for every environment $\xi$ and every surrogate in
  $\Gamma$:
  \begin{equation}
  \label{eq:purepast:surrogate-invariant}
  \begin{aligned}
  \rho_i,\xi\models\ys_\psi(\bar x)
  &\Leftrightarrow
    i>0 \text{ and }
    \bar\sigma,\xi,i-1\models\psi(\bar x)\\
  \rho_i,\xi\models\zs_\psi(\bar x)
  &\Leftrightarrow
    i=0 \text{ or }
    \bar\sigma,\xi,i-1\models\psi(\bar x).
  \end{aligned}
  \end{equation}
  The case $i=0$ follows immediately from $\phi_0$.

  As for the inductive step, suppose the invariant holds at some $0\le i<n$.
  We show that~\cref{eq:purepast:surrogate-invariant} holds for $i+1$.
  First, we prove the following claim.
  \begin{claim}
    For all $\psi \in \closure(\phi)$, it holds that:
    \begin{equation}
    \label{eq:purepast:current-evaluation}
      \rho_i \oplus \letter_i,\xi\models\snf_s(\psi)
      \quad\Leftrightarrow\quad
      \bar\sigma,\xi,i\models\psi.
    \end{equation}
  \end{claim}
  \begin{claimproof}
    We proceed by induction on the structure of $\psi\in\closure(\phi)$.
    Atoms, Boolean connectives, and first-order quantifiers are evaluated
    in the usual way, since $\rho_i \oplus \letter_i$
    agrees with $\letter_i$ on $\Sigma$ and has the same domains.

    For $\psi=\ltl{Y\psi_1}$, it holds that $\bar\sigma,\xi,i \models \psi$
    iff $i>0$ and $\bar\sigma,\xi,i-1 \models \psi_1$.  If $i=0$, then both
    $\bar\sigma,\xi,i \models \psi$ and $\rho_i \oplus
    \letter_i,\xi\models\snf_s(\psi)$ are false.  If $i>0$, by the
    inductive hypothesis of~\cref{eq:purepast:surrogate-invariant}, $\bar\sigma,\xi,i-1 \models \psi_1$ is equivalent to $\rho_{i}
    \oplus \letter_{i},\xi \models \ys_{\psi_1}$. Since $\snf_S(\psi)
    = \ys_{\psi_1}$, this is also equivalent to $\rho_{i} \oplus
    \letter_{i},\xi \models \snf_S(\psi)$.  The case for $\psi
    = \ltl{Z\psi_1}$ is analogous.

    For $\psi=\ltl{\psi_1 S \psi_2}$, we have $\bar\sigma,\xi,i \models
    \psi$ iff either $\bar\sigma,\xi,i \models \psi_2$ or $\bar\sigma,\xi,i
    \models \psi_1 \land \ltl{Y \psi}$.  The case $\bar\sigma,\xi,i \models
    \psi_2$, by inductive hypothesis, is equivalent to $\rho_i \oplus
    \letter_i,\xi \models \snf_S(\psi_2)$.  The case $\bar\sigma,\xi,i
    \models \psi_1 \land \ltl{Y \psi}$ is equivalent to $\bar\sigma,\xi,i
    \models \psi_1$ and $\bar\sigma,\xi,i-1 \models \psi$ if $i>0$, and to
    false otherwise. By inductive hypothesis
    of~\cref{eq:purepast:current-evaluation,eq:purepast:surrogate-invariant},
    respectively, this means that $\rho_i \oplus \letter_i,\xi \models
    \snf_S(\psi_1)$ and $\rho_i \oplus \letter_i,\xi \models \ys_{\psi}$.
    To conclude, since $\snf_S(\psi) = \snf_S(\psi_2)\lor
    (\snf_S(\psi_1)\land\ys_\psi)$, this is equivalent to $\rho_i \oplus
    \letter_i,\xi \models \snf_S(\psi)$.  The case $\theta=\ltl{\alpha
    T\beta}$ is analogous.
  \end{claimproof}

  Now, for every $s_\psi\in\YS\cup\ZS$, the transition condition
  and~\cref{eq:purepast:current-evaluation} yield
  \begin{align*}
    \rho_{i+1},\xi\models s_\psi
    &\Leftrightarrow
      \rho_i \oplus \letter_i,\xi\models\snf_S(\psi)\\
    &\Leftrightarrow
      \bar\sigma,\xi,i\models\psi.
  \end{align*}
  which is precisely~\cref{eq:purepast:surrogate-invariant} at $i+1$.

  Finally, since $\phi_F=\ys_\phi$, the invariant at $i=n$ gives
  $\rho_n\models\phi_F$ iff $\bar\sigma,n-1\models\phi$.  Hence the unique
  run of $\bar\sigma$ is accepting if and only if $\bar\sigma$ satisfies
  $\phi$ under the pure-past semantics. The empty word is rejected because
  $\phi_0$ entails $\neg\ys_\phi$, whereas $\phi_F=\ys_\phi$. Therefore,
  $\lang_\theory(\autom(\phi))=\lang_\theory(\phi)$.
\end{proof}

\pastdeterminism*
\begin{proof}
  Let $\autom(\phi)= \seq{\Sigma,\Gamma,\phi_0,\phi_T,\phi_F}$ be the
  first-order automaton obtained by~\cref{enc:past:encoding}.  Since $\phi$
  is pure-past, $\Gamma$ consists only of sort symbols and the surrogate
  predicates in $\YS\cup\ZS$.

  We first consider the initial condition.  Fix an arbitrary
  $\sorts(\Gamma)$-structure.  The condition $\phi_0$ requires every
  $\ys_\psi$ to be interpreted as the empty relation and every $\zs_\psi$
  as the full relation over the corresponding sorts.  These assignments
  define a $\Gamma$-expansion satisfying $\phi_0$, and no other expansion
  can satisfy $\phi_0$, since the interpretation of every non-sort symbol
  in $\Gamma$ is prescribed. Thus $\phi_0$ is complete
  (\cref{def:complete:phi0}), and also deterministic (the non-sort symbols
  of $\Gamma$ are implicitly defined by its sort symbols under $\phi_0$,
  \cref{def:determinism}).

  We next consider the transition relation, which has precisely the form
  \begin{align*}
    \phi_T =
    \bigwedge_{s_\psi\in\YS\cup\ZS}
    \forall\bar x \,
    \bigl(
      s'_\psi(\bar x)
      \leftrightarrow \snf_S(\psi(\bar x))
    \bigr)
  \end{align*}
  Because $\phi$ is pure-past, each $\snf_S(\psi)$ is a formula over
  $\Gamma\cup\Sigma$: in particular, it contains no primed surrogates.

  Let $\rho \oplus \letter$ be an arbitrary $(\Gamma\cup\Sigma)$-structure.
  Define a $(\Gamma\cup\Sigma\cup\Gamma')$-expansion $\rho \oplus \letter
  \oplus \rho'$ of $\rho \oplus \letter$ by interpreting each primed
  surrogate as
  \begin{equation*}
    (s'_\psi)^{\rho \oplus \letter \oplus \rho'} \coloneqq 
    \set{
      \bar a \suchthat
      \rho \oplus \letter,\xi[\bar x\mapsto\bar a]
      \models \snf_S(\psi(\bar x))
    },
  \end{equation*}
  where $\bar a$ ranges over tuples of the appropriate sorts and $\xi$ is
  any environment.

  By construction, $\rho \oplus \letter \oplus \rho'\models\phi_T$, so
  $\autom(\phi)$ has a total transition function
  (\cref{def:total:trans:fun}).  Conversely, every expansion of $\rho
  \oplus \letter$ satisfying $\phi_T$ must assign exactly the above
  relation to each $s'_\psi$, and therefore must coincide with $\rho \oplus
  \letter$.  Hence every $(\Gamma\cup\Sigma)$-structure has exactly one
  expansion satisfying $\phi_T$.  It follows that $\phi_T$ is total and
  that the primed non-rigid state symbols are implicitly defined by
  $\Gamma\cup\Sigma$ under $\phi_T$.

  Therefore, $\autom(\phi)$ is complete and deterministic
  by~\cref{def:completeness,def:determinism}.
\end{proof}

\propunraveling*
\begin{proof}
  $(\Rightarrow)$
  Suppose that $\lang_\theory(\autom) \neq \emptyset$ and let $\bar\sigma$
  be a $\theory$-word in $\lang_\theory(\autom)$. We divide in cases
  depending on whether $\bar\sigma = \epsilon$.  If $\bar\sigma
  = \epsilon$, then this necessarily means that there exists
  a $\Gamma$-structure $\rho_0$ such that $\rho_0 \models \phi_0 \land
  \phi_F$, that is, $\rho_0^0 \models \sem{\autom}^F_0$ and
  $\sem{\autom}^F_0$ is satisfiable.
  If instead $\bar\sigma = \seq{\letter_0,\dots,\letter_{k-1}}$, for some
  $k>0$, then there is a $\emptyset$-word $\bar\rho
  = \seq{\rho_0,\dots,\rho_{k}}$ of $\Gamma$-structures such that:
  \begin{enumerate*}[label=(\roman*)]
    \item $S^{\bar\rho}=S^{\bar\sigma}$ for all $S\in\sorts(\Sigma)$;
    \item $\rho_0\models\phi_0$;
    \item $\rho_i\oplus\letter_i\oplus\rho_{i+1}'\models \phi_T$ for all
      $0\le i < k$; and
    \item $\rho_k \models \phi_F$.
  \end{enumerate*}
  By the renaming of the non-rigid symbols performed by $\sem{\autom}^F_k$,
  we obtain the following structure:
  \[
    \widehat{\rho} \coloneqq \rho_0^0 \oplus \letter_0^0 \oplus \dots \oplus
    \rho^{k-1}_{k-1} \oplus \letter^{k-1}_{k-1} \oplus
    \rho^{k}_{k}
  \] 
  Since the renamed signatures overlap only on rigid symbols and their
  interpretation agree along $\bar\sigma$ and $\bar\rho$, and since the
  state and word structures agree on the common sort symbols, we have that
  $\widehat{\rho}$ is a well-defined $(\Sigma^{0\ldots k-1} \cup
  \Gamma^{0\ldots k})$-structure.

  Moreover, $\widehat{\rho}$ satisfies $\phi_0^0$, satisfies
  $\bigwedge_{i=0}^{k-1} \phi_T^i$ and satisfies $\phi^k_F$, and thus is
  a model of $\sem{\autom}^F_k$. Since $\bar\sigma$ is
  a $\theory$-word, $\widehat{\rho}$ satisfies also $\theory^{0\ldots
  k-1}$. Therefore $\sem{\autom}^F_k$ is $\theory^{0\ldots
  k-1}$-satisfiable.

  $(\Leftarrow)$
  Suppose that $\sem{\autom}^F_0$ is satisfiable. By definition,
  $\sem{\autom}^F_0 \equiv \phi^0_0 \land \phi^0_F$. This means that there
  exists a $\Gamma^0$-structure $\rho^0_0$ such that $\rho^0_0 \models
  \phi^0_0$ and $\rho^0_0 \models \phi^0_F$. This is equivalent to $\rho_0
  \models \phi_0$ and $\rho_0 \models \phi_F$, where $\rho_0$ is the
  $\Gamma$-structure obtained from $\rho^0_0$ by renaming each non-rigid
  symbol $s^0$ with $s$. Thus, the $\emptyset$-word $\seq{\rho_0}$ is an
  accepting run of $\autom$ induced by $\epsilon$ and $\epsilon \in
  \lang_\theory(\autom)$, having that $\lang_\theory(\autom) \neq
  \emptyset$.

  Finally, suppose that $\sem{\autom}^F_k$ is $\theory^{0\ldots
  k-1}$-satisfiable for some $k\ge 1$. By definition of $\sem{\autom}^F_k$
  and of $\theory^{0\ldots k-1}$-satisfiability, there exists
  a $(\Sigma^{0\ldots k-1} \cup \Gamma^{0\ldots k})$-structure
  $\widehat{\rho}$ satisfying $\theory^{0\ldots k-1}$ and $\phi^0_0 \land
  \bigwedge_{i=0}^{k-1} \phi_T^i \land \phi^k_F$. Starting from
  $\widehat{\rho}$, we define:
  \begin{enumerate}
    \item the $\Gamma^{i}$-structure $\widehat{\rho}_i$, which is the
      reduct of $\widehat{\rho}$ with respect to $\Gamma^{i}$, for each
      $0\le i \le k$; and
    \item the $\Sigma^{i}$-structure $\widehat{\letter}_i$, which is the
      reduct of $\widehat{\rho}$ with respect to $\Sigma^{i}$, for each
      $0\le i < k$.
  \end{enumerate}
  It is easy to see that:
  \begin{enumerate*}[label=(\roman*)]
    \item $\widehat{\letter}_i \models \theory^i$, for each $0\le i< k$;
    \item $\widehat{\rho}_0 \models \phi_0^0$;
    \item $\widehat{\rho}_k \models \phi_F^k$; and
    \item $\widehat{\rho}_i \oplus \widehat{\letter}_i \oplus
      \widehat{\rho}_{i+1} \models \phi_T^{i}$, for each $0\le i < k$.
  \end{enumerate*}
  For each $0\le i < k$, let $\letter_i$ be the $\Sigma$-structure obtained
  from $\widehat{\letter}_i$ by renaming each non-rigid symbol $s^i$ with
  $s$. For each $0 \le i \le k$, let $\rho_i$ be the $\Gamma$-structure
  obtained from $\widehat{\rho}_i$ by renaming each non-rigid symbol $s^i$
  with $s$. It holds that:
  \begin{enumerate}
    \item $\bar\sigma = \seq{\letter_0,\dots,\letter_{k-1}}$ is
      a $\theory$-word, because $\letter_i \models \theory$, for each $0\le
      i< k$, and each rigid symbol has the same interpretation across all
      positions;
    \item $\bar\rho = \seq{\rho_0,\dots,\rho_k}$ is a $\emptyset$-word that
      satisfies all the conditions for being an accepting run
      (\cref{def:run,def:acceptance}) because:
      \begin{itemize}
        \item since rigid symbols are not renamed, all the structures
          $\rho_i$ agree on the rigid symbols of $\Gamma$, and hence
          $\bar\rho=\seq{\rho_0,\ldots,\rho_k}$ is an $\emptyset$-word;
        \item since $\Gamma$ and $\Sigma$ share the same sort symbols and
          all the structures are reducts of $\widehat{\rho}$, we have
          $S^{\bar\rho}=S^{\bar\sigma}$ for every $S\in\sorts(\Sigma)$;
        \item $\rho_0 \models \phi_0$;
        \item $\rho_k \models \phi_F$; and
        \item $\rho_i \oplus \letter_i \oplus \rho'_{i+1} \models \phi_T$, for
          each $0\le i< k$.
      \end{itemize}
  \end{enumerate}
  It follows that the $\emptyset$-word $\bar\rho
  = \seq{\rho_0,\dots,\rho_k}$ is an accepting run induced by the
  $\theory$-word $\bar\sigma = \seq{\letter_0,\dots,\letter_{k-1}}$.
  Therefore, $\bar\sigma \in \lang_\theory(\autom)$ and
  $\lang_\theory(\autom) \neq \emptyset$.
\end{proof}

\thmdecisionfinitecontrol*
\begin{proof}[Sketch]
  Since the number of states of a finite-control automaton is finite, every
  non-accepting long-enough run either ends in a state with no successors or
  contains two repetitions of the same state. Consider the formula
  $\phi^k_P\equiv
  \bigvee_{i=0}^{k}\bigvee_{j=i+1}^{k-1}\bigwedge_{p\in\Gamma}(p^i \iff p^j)$.
  Intuitively, $\phi^k_P$ states the existence of a loop through two identical
  states.  \Cref{algo:semidecision} can be turned into a decision procedure by
  adding the following test at each iteration: if $\sem{\autom}_k\to\phi_P^k$ is
  \emph{valid} (\ie every run of $\autom$ of length $k$ contains at least two
  repetitions of the same state), we reply that $\lang_\theory(\autom)$ is
  \emph{empty}.
\end{proof}

\subsection*{Proof of \cref{thm:monodic:decidability}}

To show it, we need some preliminary work. Let
$\autom=\seq{\Sigma,\Gamma,\phi_0,\phi_T,\phi_F}$ be a monadic automaton. Assume
\emph{w.l.o.g.}\ that $\phi_0$, $\phi_T$ and $\phi_F$ are written in
\emph{negated normal form}. For conciseness, we consider here the case where
$\Sigma$ is \emph{monosorted} and has no constants at all, but what follows can
be extended to the case of multiple sorts and \emph{rigid} constants.

A \emph{type} is a subset $t\subseteq\Gamma$ of the monadic predicates of the
state signature. Intuitively, a type represents all the elements in a
$\Gamma$-structure that satisfy all and only the type's predicates. An
\emph{abstract state} is a set $s\subseteq 2^\Gamma$ of types. Note that the
number of possible abstract states is finite. Intuitively, an abstract state
abstracts finitarily a $\Gamma$-structure by representing the existence or
absence of elements satisfying a given collection of types. This notion is
called \emph{state candidate} in \cite{GabbayKWZ03}, and can be seen as a
weakening of Robinson diagrams~\cite{ChangKeisler}. It is useful to define a
sentence $\gamma_s$ that captures the above intuition about an abstract state
$s$, that is, $\gamma_s\equiv \bigwedge_{t\in s} \exists
x.\gamma_t(x)\land\bigwedge_{t\not\in s}\neg\exists x.\gamma_t(x)$ where
$\gamma_t(x)\equiv \bigwedge_{p \in t} p(x)\land \bigwedge_{p \not\in t} \neg
p(x)$.

The core idea of the construction below is the following. Suppose, for example,
that the current state $\rho$ is summarised by an abstract state $s=\set{t_1,
t_2}$ where $t_1=\set{p(x), q(x)}$ and $t_2=\set{p(x)}$. This, intuitively,
means that in the state there exists at least an element satisfying both $p(x)$
and $q(x)$, but all the elements satisfy at least $p(x)$. Then, under this
assumption, if we encounter a sentence such as $\exists x.\neg p(x)$, we know
that it will not hold over $\rho$. Similarly, if we encounter $\forall x. q(x)$,
we can know it will not hold over $\rho$, because the elements of type $t_2$ do
\emph{not} satisfy $q(x)$. A further detail to consider is that, to reason
about $\phi_T$ as above, we need to abstract both the source and the
destination state.

The above intuition is realised as follows. An \emph{assumption} is a function
$h$ from any first-order variable $x$ to a pair of types $h(x)=(t_1,t_2)$,
describing $x$ now and at the next state. Given two states $s_1$ and $s_2$, an
assumption $h$, and a $(\Gamma\cup\Sigma\cup\Gamma')$-formula $\psi$, we define
the formula $\psi\vert^h_{s_1\to s_2}$ inductively as follows: 
\begin{enumerate}
  \item if $p\in\Gamma$, then $p(x)\vert^h_{s_1\to s_2}\equiv \top$ if $p\in
    t_1$, or $\bot$ otherwise, where $h(x)=(t_1,t_2)$; for $p'\in\Gamma'$, we
    define $p'(x)\vert^h_{s_1\to s_2}$ similarly but using $t_2$;
  \item if $q\in\Sigma$, then 
    $q(x_1,\ldots,x_k)\vert^h_{s_1\to s_2}\equiv q(x_1,\ldots,x_k)$;
  \item 
    $(\psi_1{\circ}\psi_2)\vert^h_{s_1\to s_2}{\equiv} (\psi_1\vert^h_{s_1\to
    s_2}){\circ}(\psi_2\vert^h_{s_1\to s_2})$ for $\circ\in\set{\land,\lor}$;
  \item $(Q x.\psi_1)\vert^h_{s_1 \to s_2}\equiv Q x.\bigvee_{t_1\in s_1}\bigvee_{t_2\in s_2}\psi_1\vert^{h[x\mapsto (t_1,t_2)]}_{s_1\to s_2}$ for $Q\in\set{\forall,\exists}$.
\end{enumerate}

If $\psi$ is a sentence, by $\psi\vert_{s_1\to s_2}$ we denote
$\psi\vert^h_{s_1\to s_2}$ for some arbitrary assumption $h$. We can now
proceed.
\begin{encoding}[Monadic to finite-control automaton]
  \label{enc:monadic:encoding}
  Given a monadic automaton $\autom=\seq{\Sigma,\Gamma,\phi_0,\phi_T,\phi_F}$,
  we define $\autom^*=\seq{\Sigma,\Gamma_*,\phi^*_0,\phi^*_T,\phi^*_F}$ to be
  the following automaton:
  \begin{enumerate}
    \item the state signature $\Gamma_*$ contains a \emph{proposition} $b_t$ for
      each type $t\subseteq\Gamma$;
    \item the initial condition $\phi^*_0$ is defined as the disjunction of
      $b_s\equiv \bigwedge_{t\in s}b_t\land\bigwedge_{t\not\in s}\neg b_t$ for
      each abstract state $s$ such that $\gamma_s\models\phi_0$;
    \item the final condition $\phi^*_F$ is defined in the same way as
      $\phi^*_0$, but on top of $\phi_F$;
    \item the transition relation $\phi^*_T$ is defined as:
    \begin{equation*}
      \phi^*_T\equiv 
        \smashoperator{\bigwedge_{\substack{s_1,s_2\\\text{abstract states}}}}
          \left[(b_{s_1} \land b'_{s_2})\implies \psi_T\vert_{s_1\to s_2}\right]
    \end{equation*}
  \end{enumerate}
\end{encoding}

Note that $\autom^*$ is finite-control. We need now one last lemma.
\begin{lemma}
  \label{lemma:gammastate}
  Let $\Gamma$ be the state signature of a \emph{monosorted} monadic automaton.
  Then, for each $\Gamma$-structures $\mu$ and $\rho$, and any abstract state
  $s$, if both $\mu\models\gamma_s$ and $\rho\models\gamma_s$, then $\mu$ and
  $\rho$ are elementarily equivalent.
\end{lemma}
\begin{proof}
  Suppose $S$ is the only sort involved. Let $s$ be the abstract state such that
  $\mu\models\gamma_s$ and $\rho\models\gamma_s$. We define a map $f:S^\mu\to
  S^\rho$ such that for each $a\in S^\mu$ that models a type $t\in s$, $f(a)=b$
  where $b$ is an arbitrary element in $S^\rho$ that models $t$ as well. Given
  an environment $\xi$, we denote as $\xi_f$ the environment such that
  $\xi_f(x)=f(\xi(x))$ for any variable $x$.
  
  We prove a slightly more general statement by structural induction over a
  $\Gamma$-formula $\psi$: for each environment $\xi$, if $\mu,\xi\models\psi$,
  then $\rho,\xi_f\models\psi$. For the base case of an atomic formula $p(x)$,
  if $\mu,\xi\models p(x)$, then $\rho,\xi_f\models p(x)$ by definition of $f$
  and $\xi_f$. The case for $\neg p(x)$ is similar. For the inductive step, if
  $\mu,\xi\models\exists x.\psi(x)$, then there is a $\xi'$, that agrees with
  $\xi$ everywhere except only possibly on $x$, such that
  $\mu,\xi'\models\psi(x)$. Then, by the inductive hypothesis, we know
  $\rho,\xi'_f\models\psi(x)$. But if $\xi'$ agrees with $\xi$ except possibly
  on $x$, also $\xi'_f$ agrees with $\xi_f$ except possibly on $x$, by
  definition. Therefore, $\rho,\xi_f\models\exists x.\psi(x)$. Similarly, if
  $\mu,\xi\models\forall x.\psi(x)$, then we have $\mu,\xi'\models\psi(x)$ for
  all $\xi'$ that agrees with $\xi$ except possibly on $x$. But then
  $\rho,\xi'_f\models\psi(x)$ by the induction hypothesis, and therefore
  $\rho,\xi_f\models\forall x. \psi(x)$. The cases for Boolean connectives are
  simple.
\end{proof}

We can now proceed to prove the main result.

\begin{theorem}[Correctness of \cref{enc:monadic:encoding}]
  \label{thm:encoding:monadic}
  Consider a \emph{monadic} automaton
  $\autom=\seq{\Sigma,\Gamma,\phi_0,\phi_T,\phi_F}$ such that:
  \begin{enumerate}
  \item $\Sigma$ is a monosorted \emph{relational} signature with no constants;
  \item $\phi_0$, $\phi_T$, and $\phi_F$ do not use equality.
  \end{enumerate}
  Then, the automaton $\autom^*$ of \cref{enc:monadic:encoding} is equivalent to
  $\autom$.
\end{theorem}
\begin{proof}
  Let the only sort in $\Sigma$ be $S$ and let
  $\autom=\seq{\Sigma,\Gamma_*,\phi^*_0,\phi^*_T,\phi^*_F}$ as in
  \cref{enc:monadic:encoding}. For any $\Sigma$-theory $\theory$ we want to show
  that $\lang_{\theory}(\autom)=\lang_{\theory}(\autom^*)$, so we need to show
  the two directions.

  $(\subseteq)$ Let $\model=\seq{\letter_0,\ldots,\letter_{n-1}}$ be a
  $\theory$-word accepted by $\autom$ and let
  $\bar\rho=\seq{\rho_0,\ldots,\rho_n}$ be its induced accepting run. We
  construct a run in $\autom^*$ and show it is indeed an accepting run. The run
  is defined as $\bar\rho^*=\seq{\rho_0^*,\ldots,\rho_n^*}$ where
  $\rho^*_i\models b_{s_i}$ for some abstract state $s_i$ with $0\le i\le n$ if
  and only if $\rho_i\models\gamma_{s_i}$. Note that there exists and is unique
  the $s_i$ such that this can occur for each $i$, because $\bigvee_s\gamma_s$
  is valid, but $\gamma_{s_1}\land\gamma_{s_2}$ is unsatisfiable for each
  distinct abstract states $s_1$ and $s_2$. 
  
  Now suppose $i=0$. We know $\rho_0\models\gamma_{s_0}$ and
  $\rho_0\models\phi_0$, therefore, by \cref{lemma:gammastate}, we also have
  that $\gamma_{s_0}\models\phi_0$ because for any $\Gamma$-structure $\mu$, if
  $\mu\models\gamma_{s_0}$ then $\mu$ is elementarily equivalent to $\rho_0$ and
  therefore it holds that $\mu\models\phi_0$ just like $\rho_0$. Since
  $\gamma_{s_0}\models\phi_0$, by \cref{enc:monadic:encoding} we have $b_{s_0}$
  is a disjunct of $\rho^*_0$, and therefore $\rho^*_0\models\phi_0^*$. A
  similar reasoning goes to show that $\rho^*_n\models\phi^*_F$.
  
  Now for each $0<i<n$, we want to show that
  $\rho^*_i\oplus\letter_i\oplus\rho^*_{i+1}\models\phi^*_T$. The definition of
  $\phi^*_T$ conjuncts a set of implications, one for each combination of
  current and next abstract state. So let $s_i$ and $s_{i+1}$ be the abstract
  states corresponding to $\rho^*_i$ and $\rho^*_{i+1}$. Following the
  definition of $\phi^*_T$ we have to show that if
  $\rho_i\oplus\letter_i\oplus\rho'_{i+1}\models\psi_T$, then
  $\rho^*_i\oplus\letter_i\oplus\rho^*_{i+1},\xi\models\psi_T\vert^h_{s_i\to
  s_{i+1}}$ for some assumption $h$. But notice that by definition,
  $\psi_T\vert^h_{s_i\to s_{i+1}}$ does not mention anything from $\Gamma$ and
  $\Gamma'$, so it suffices to show that $\letter_i\models
  \psi_T\vert^h_{s_i\to s_{i+1}}$. 
  
  We actually show the more general statement that for any environment $\xi$ and
  any formula $\psi$, if $\rho_i\oplus\letter_i\oplus\rho'_{i+1},\xi\models\psi$,
  then $\letter_i,\xi\models\psi\vert^{h_\xi}_{s_i\to s_{i+1}}$ where $h_\xi$ is
  the assumption such that, if $h_\xi(x)=(t_i,t_{i+1})$, then $t_i$ is the type
  of $\xi(x)$ in $\rho_i$ and $t_{i+1}$ is the type of $\xi(x)$ in $\rho_{i+1}$.
  
  We do that by structural induction over $\psi$. For the base case of an atomic
  formula, $p(x_1,\ldots,x_k)$, if $p\in\Sigma$ then
  $p(x_1,\ldots,x_k)\vert_{s_i\to s_{i+1}}$ is just $p(x_1,\ldots,x_k)$, and we
  know $\rho_i\oplus\letter_i\oplus\rho'_{i+1},\xi\models p(x_1,\ldots,x_k)$. But
  this only depends on $\letter_i$, so we have $\letter_i,\xi\models
  p(x_1,\ldots,x_k)$ as well. If instead $p\in\Gamma$, then it is monadic, and
  the atomic formula is actually just $p(x)$. We know that
  $\rho_i\oplus\letter_i\oplus\rho'_{i+1},\xi\models p(x)$, so if
  $h_\xi(x)=(t_i,t_{i+1})$, then we know $p\in t_i$. Then
  $p(x)\vert^{h_\xi}_{s_1\to s_2}\equiv\top$ by definition, so trivially
  $\letter_i,\xi\models p(x)\vert^{h_\xi}_{s_1\to s_2}$. A similar argument
  works for $\neg p(x)$, $p'(x)$, $\neg p'(x)$.

    For the inductive case, we distinguish the kind of formula at hand:
  \begin{enumerate}
    \item if $\rho_i\oplus\letter_i\oplus\rho'_{i+1},\xi\models\forall x.\psi$, then
      for any $\xi'$ agreeing with $\xi$ except for $x$, we have
      $\rho_i\oplus\letter_i\oplus\rho'_{i+1},\xi'\models \psi$. Then, for the
      inductive hypothesis, we know $\letter_i,\xi'\models
      \psi\vert^{h_{\xi'}}_{s_1\to s_2}$. Now, notice that $h_{\xi'}$ agrees
      with $h_\xi$ everywhere except for $x$, so let
      $h_{\xi'}(x)=(t_i,t_{i+1})$. Since we chose $s_i$ and $s_{i+1}$ as the
      abstract states of $\rho_i$ and $\rho_{i+1}$, for sure $t_i\in s_i$ and
      $t_{i+1}\in s_{i+1}$. Hence we have:
      \begin{align*}
        \letter_i,\xi'\models {} &\bigvee_{t_i\in s_i}\bigvee_{t_{i+1}\in s_{i+1}}\psi\vert^{h_\xi[x\mapsto(t_i,t_{i+1})]}_{s_i\to s_{i+1}}
      \shortintertext{and therefore:}
      \letter_i,\xi\models {} &\forall x.\bigvee_{t_i\in s_i}\bigvee_{t_{i+1}\in s_{i+1}}\psi\vert^{h_\xi[x\mapsto(t_i,t_{i+1})]}_{s_i\to s_{i+1}}
      \end{align*}
      which means exactly, by definition, $\letter_i,\xi\models(\forall
      x.\psi)\vert^{h_\xi}_{s_i\to s_{i+1}}$. The reasoning with the existential
      quantifier is identical.
       
    \item if $\psi$ is a conjunction or a disjunction the induction is trivial.
  \end{enumerate}

  $(\supseteq)$ On the converse, let
  $\model=\seq{\letter_0,\ldots,\letter_{n-1}}$ be a $\theory$-word accepted by
  $\autom^*$ and let $\bar\rho^*=\seq{\rho^*_0,\ldots,\rho^*_n}$ be its induced
  accepting run. Then we define a run $\bar\rho=\seq{\rho_0,\ldots,\rho_n}$ over
  $\autom$ and we show that it is accepting. The run is defined as follows. For
  each $0\le i\le n$ we let $\rho_i$ be any $\Gamma$-structure such that
  $\rho_i\models\gamma_{s_i}$ for some abstract state $s_i$ if and only if
  $\rho^*_i\models b_{s_i}$. As in the other direction, note that the $s_i$ for
  which this happens exists and is unique. We show now that this is indeed a run
  of $\autom$ and that it is accepting.

  We start showing that $\rho_0\models\phi_0$. We know that
  $\rho^*_0\models\phi^*_0$, hence $\rho^*_0\models b_{s_0}$ and we know
  $\gamma_{s_0}\models \phi_0$. Since $\rho_0\models\gamma_{s_0}$, it follows
  directly that $\rho_0\models\phi_0$. The reasoning to show that
  $\rho_n\models\phi_F$ is similar.

  For $0 < i < n$, we show that
  $\rho_i\oplus\letter_i\oplus\rho'_{i+1}\models\phi_T$, given that
  $\rho^*_i\oplus\letter_i\oplus\rho^*_{i+1}\models\phi^*_T$. We know that
  $\rho^*_i\models b_{s_i}$ and $\rho^*_{i+1}\models b_{s_{i+1}}$, therefore by
  definition of $\phi^*_T$, it follows that
  $\rho^*_i\oplus\letter_i\oplus\rho^*_{i+1}\models \phi_T\vert^{h_i}_{s_i\to
  s_{i+1}}$ for some hypothesis $h_i$. But, since $\phi_T\vert^{h_i}_{s_i\to
  s_{i+1}}$ does not mention anything from $\rho^*_i$ nor $\rho^*_{i+1}$, we can
  just say that $\letter_i\models\phi_T\vert^{h_i}_{s_i\to s_{i+1}}$.
  
  We show the claim by showing the more general statement that for each formula
  $\psi$ and each environment $\xi$, if
  $\letter_i,\xi\models\psi\vert^{h_i}_{s_i\to s_{i+1}}$, then there is a $\xi'$
  such that $\rho_i\oplus\letter_i\oplus\rho'_{i+1},\xi'\models\psi$. We do that by
  structural induction on $\psi$.

  For the base case of an atomic formula, if $\letter_i,\xi\models
  p(x_1,\ldots,x_k)\vert^{h_i}_{s_i\to s_{i+1}}$, then
  $p(x_1,\ldots,x_k)\vert^{h_i}_{s_i\to s_{i+1}}$ is just $p(x_1,\ldots,x_k)$,
  so the claim follows trivially. If instead $p\in\Gamma$, then it is monadic
  and thus the atomic formula is just $p(x)$. In this case, let
  $h_i(x)=(t_i,t_{i+1})$. Then, $p(x)\vert^{h_i}_{s_i\to s_{i+1}}$ is either
  $\top$, if $p\in t_i$, or $\bot$, if $p\not\in t_i$. But it cannot be $\bot$,
  because we cannot have that $\letter_i,\xi\models\bot$. Therefore, $p\in t_i$.
  Since $t_i\in s_i$, by how we defined the run $\bar\rho$, we know there is an
  element $a\in p^{\rho_i}$, therefore we set $\xi'(x)=a$ and we have
  $\rho_i\oplus\letter_i\oplus\rho'_{i+1},\xi'\models p(x)$. A similar argument works
  for $\neg p(x)$, $p'(x)$, $\neg p'(x)$.

  For the inductive case, we distinguish the kind of formula at hand:
  \begin{enumerate}
    \item if $\letter_i,\xi\models(\forall x.\psi)\vert^{h_i}_{s_i\to s_{i+1}}$,
      then $\letter_i,\xi\models \forall x.\bigvee_{t_i\in
      s_i}\bigvee_{t_{i+1}\in
      s_{i+1}}\psi\vert^{h_i[x\mapsto(t_i,t_{i+1})]}_{s_i\to s_{i+1}}$, and then
      for some $t_i\in s_i$ and $t_{i+1}\in s_{i+1}$ we have
      $\letter_i,\xi'\models \psi\vert^{h_i'}_{s_i\to s_{i+1}}$ for some $h_i'$
      and for all $\xi'$ that agree with $\xi$ except possibly on $x$. Then by
      the inductive hypothesis, we know there is a $\xi''$ such that
      $\rho_i\oplus\letter_i\oplus\rho'_{i+1},\xi''\models \psi$, hence
      $\rho_i\oplus\letter_i\oplus\rho'_{i+1},\xi'''\models \forall x.\psi$ for any
      $\xi'''$ that agrees with $\xi''$ everywhere except possibly on $x$.
    \item if $\psi$ is a conjunction or a disjunction the induction is trivial.\qedhere
  \end{enumerate}
\end{proof}

The first restriction on $\Sigma$ can be lifted by suitably extending the notion
of abstract state to account for constants and multiple sorts. Then, note that
\cref{prop:monodic} follows directly from
\cref{thm:decision:finite-control,thm:encoding:monadic}.

%% file: sections/b.encodings.tex

\newcommand{\m}[1]{\mathsf{#1}}
\newcommand{\mc}[1]{\mathcal{#1}}
\newcommand{\obar}[1]{\makebox[0pt]{$\phantom{#1}\overline{\phantom{#1}}$}#1}

\newcommand{\mydds}[1][$\Pi$]{{#1}-DMT\xspace}

\newcommand{\CC}{\mathit{Constr}} 
\newcommand{\Sort}{\mc S} 
\newcommand{\PP}{\mc P} 
\newcommand{\FF}{\mc F} 
\newcommand{\MM}{\mc F} 
\newcommand{\BB}{\mc B} 
\newcommand{\dmttuple}{\langle \Pi, \mc{V}, \mc{I}, Tr\rangle}
\newcommand{\TT}{\mathit{Thr}} 
\newcommand{\goto}[1]{\mathrel{\raisebox{-2pt}{$\xrightarrow{#1}$}}}
\newcommand\D{\mathcal{D}}
\newcommand\A{\mathsf{A}}
\newcommand\enc{\mathrm{enc}}

\section{Encodings of different formalisms}
\label{app:encodings}

\subsection*{Symbolic Automata}

\subsubsection*{Definition.}

In this section we recall the definition of \emph{symbolic finite automata}
(s-FA) as provided \eg in \cite{DAntoniV17}. Symbolic automata are finite-state
automata that read characters of large or infinite alphabets defined by
\emph{effective Boolean algebras}.
\begin{definition}[Effective Boolean algebra]
    \label{def:boolean}
    An \emph{effective Boolean algebra} is a tuple 
    $\A=\seq{\D,\Psi,\sem{\_},\bot,\top,\lor,\land,\neg}$ where:
    \begin{enumerate}
        \item $\D$ is a set of \emph{domain elements};
        \item $\Psi$ is a set of \emph{predicates} closed under Boolean connectives, with $\top,\bot\in\Psi$;
        \item $\sem{\_}:\Psi\to 2^\D$ is a \emph{denotation function} such that for all $\phi,\psi\in\Psi$ we have:
        \begin{itemize}
            \item $\sem{\bot}=\emptyset$;
            \item $\sem{\top}=\D$;
            \item $\sem{\phi\lor\psi}=\sem{\phi}\cup\sem{\psi}$;
            \item $\sem{\phi\land\psi}=\sem{\phi}\cap\sem{\psi}$; and
            \item $\sem{\neg\phi}=\D\setminus\sem{\psi}$;
        \end{itemize}
        \item for all $\phi\in\Psi$, telling whether $\sem{\phi}\ne\emptyset$ is \emph{decidable}.
    \end{enumerate}
\end{definition}

Elements of $\D$ are called \emph{characters} and symbolic automata read finite
words from $\D^*$.
\begin{definition}[Symbolic finite automata]
    A \emph{symbolic finite automaton} (s-FA) is a tuple 
        $M=(\A, Q, q_0,F,\Delta)$ where:
    \begin{enumerate}
        \item $\A$ is an effective Boolean algebra;
        \item $Q$ is a \emph{finite} set of states;
        \item $q_0\in Q$ is the \emph{initial state};
        \item $F\subseteq Q$ is the set of \emph{final states}; and
        \item $\Delta\subseteq Q\times\Psi_\A\times Q$ is a finite set of
            \emph{transitions}.
    \end{enumerate}
\end{definition}

We do not repeat the definition of the semantics of s-FAs, that is defined in
\cite{DAntoniV17}. Intuitively, s-FAs define transitions from a state to another
by means of elements of $\Psi_\A$ which denote subsets of $\D$. Whether an
element $d\in\D$ is read and the automaton is in a state $q\in Q$, it
nondeterministically transitions to $q'\in Q$ if there is a transition
$(q,\psi,q')\in\Delta$ such that $d\in\sem{\psi}$. The meaning of initial and
final states is the same as in common finite state automata.

\subsubsection*{Encoding.}

We show now that s-FAs can be encoded into finite-control first-order automata.

The first thing to settle is how to represent the words of elements of a Boolean
algebra $\A=\seq{\D,\Psi,\sem{\_},\bot,\top,\lor,\land,\neg}$ with first-order
structures over a signature $\Sigma$. One may be tempted to define a predicate
in $\Sigma$ for each predicate in $\Psi$, but by \cref{def:boolean} $\Psi$ may
also be infinite, while in first-order automata we want a finite signature
$\Sigma$. However, note that only a finite number of predicates from $\Psi$ will
be concretely mentioned in $\Delta$.

Following this observation, let $P$ be the set of predicates $\psi$ that appear
in at least some transition $(q,\psi,q')\in\Delta$. We define $\A_*$ as the
\emph{subalgebra} of $\A$ generated by $P$, that is, the effective Boolean
algebra $\A_*=\seq{\D,\Psi,\sem{\_},\bot,\top,\lor,\land,\neg}$ where $\Psi_*$
is the \emph{closure} by Boolean operators $\lor$, $\land$, and $\neg$ of $P$.

It is easy to see that any s-FA over $\A$ is equivalent to the same s-FA over
$\A_*$, with the advantage that $\A_*$ is finitely generated.

Now, given an s-FA $M=(\A, Q, q_0,F,\Delta)$, we turn it into $M_*=(\A_*, Q,
q_0,F,\Delta)$, and then we encode it into a first-order automaton
$\autom(M_*)=\seq{\Sigma,\Gamma,\phi_0,\phi_T, \phi_F}$ as follows:
\begin{enumerate}
    \item $\Sigma$ contains a single \emph{non-rigid} constant $c$ of a fixed
        sort $S$, and a \emph{rigid} unary predicate $p_\psi$ over $S$ for each
        $\psi\in\Psi_*$;
    \item $\Gamma$ is made of $n=\ceil{\log_2(|Q|)}+1$ propositions $b_0,\ldots,b_{n-1}$ where:
    \begin{itemize}
        \item we assume an injective function $\enc:Q\to 2^\Gamma$ that provides
            the set $\enc(q)$ of propositions true in a binary encoding of $q\in
            Q$
        \item for each $q\in Q$ we define 
            $\phi_q\equiv\bigwedge_{b\in\enc(q)} b\land\bigwedge_{b\not\in\enc(q)}\neg b$;
    \end{itemize}
    \item the initial condition is $\phi_0\equiv \phi_{q_0}$;
    \item the acceptance condition is $\phi_F\equiv\bigvee_{q\in F}\phi_q$;
    \item the transition relation is:
    \begin{equation*}
        \phi_T\equiv \bigvee_{(q,\psi,q')\in\Delta} (\phi_{q} \land p_\psi(c) \land \phi_{q'})
    \end{equation*}
\end{enumerate}

Now, a word $\bar d=\seq{d_0,\ldots,d_{n-1}}\in\D^*$ can be represented by a
$\emptyset$-word $\model=\seq{\letter_0,\ldots,\letter_{n-1}}$ over $\Sigma$
where for each $0\le i < n$ we have $S^{\letter_i}=\D$ and
$p_\psi^{\letter_i}=\sem{\psi}$. Conversely, from a first-order word
$\model=\seq{\letter_0,\ldots,\letter_{n-1}}$ we define a word $\bar
d=\seq{d_0,\ldots,d_{n-1}}\in\D^*$ where for each $0\le i < n$ we set
$d_i=f(c^{\letter_i})$ for some injective mapping $f:S^{\letter_i}\to \D$ such
that for each $a\in S^{\letter_i}$ we have $a\in p_\psi^{\letter_i}$ iff
$f(a)\in\sem{\psi}$. One can easily, albeit tediously, prove the following.
\begin{theorem}
    A word $\bar d$ is accepted by an s-FA $M$ if and only if the corresponding
    $\emptyset$-word $\model$ is accepted by $\autom(M_*)$, and \viceversa.
\end{theorem}

\subsection*{Data-Aware Processes Modulo Theories.}

\subsubsection*{Definition.}
In this section, we recall the definition of DMTs from \cite{GianolaMW24}. Let
$\Pi=\langle\Sort, \PP, \FF, \mc{V}, \mc{U}\rangle$ be a first-order signature,
where $\Sort$ are sorts, $\PP$ and $\FF$ are respectively predicate and function
symbols, $\mc{V}$ and $\mc{U}$ are disjoint sets of variables. The variables in
$\mc{V}$ are called \emph{data variables} and are employed to formalize control
variables of a process that store data from some domain, whereas the variables
in $\mc{U}$ are used for quantification. For a set of variables $Z$, a
\emph{$\Pi$-constraint} $c$ over $Z$ is of the form $\exists u_1, \dots,
u_l.\phi$ such that $u_1, \dots, u_l\in U$, $\phi$ is a conjunction of
$\Pi$-literals, and all free variables in $c$ are in $Z$; the set of all such
constraints is denoted $\CC_\Sigma(Z)$. For each data variable $v\in V$, let
$v^r$ and $v^w$ be two annotated variables of the same sort, and set $\mc{V}^r =
\{v^r \mid v\in \mc{V}\}$ and $\mc{V}^w = \{v^w \mid v\in \mc{V}\}$. These
copies of $V$ are called the \emph{read} and \emph{write} variables; they will
denote the variable values before and after a transition of the system,
respectively.

\begin{definition}
\label{def:dds}
A \emph{data-aware process modulo theories} over $\Pi$ (\mydds for short) is
a labelled transition system $\BB = \dmttuple$, where:
\begin{itemize}
\item $\Pi$ is a first-order signature,
\item $\mc{V}$ is the finite, nonempty set of data variables in $\Pi$;
\item the transition formulae $Tr(\mc{V}^r,\mc{V}^w)$
are a set of constraints in $\CC_\Pi(\mc{V}^r\cup \mc{V}^w)$;
\item $\mc{I}\colon \mc{V}\to \FF_0$, called initial function, 
initializes variables,
where $\FF_0$ is the set of $\Pi$-constants.
\end{itemize}
\end{definition}

We then recall the semantics for {\mydds}s.
For a $\Pi$-theory $\TT$ and a model $M\in \TT$,
a \emph{state} of a \mydds $\BB$ is an assignment $\alpha\colon \mc{V} \to |M|$.
A \emph{guard assignment} $\beta$ is a function $\beta\colon \mc{V}^r \cup \mc{V}^w \to |M|$.
As defined next, a transition $t$ can transform a state $\alpha$ into a new state $\alpha'$, updating the variable values in agreement with $t$,
while variables that are not explicitly written keep their previous value as per $\alpha$.

\begin{definition}
A \mydds $\BB = \dmttuple$ admits a \emph{$\TT$-step} from state $\alpha$ to 
$\alpha'$ via transition $t \in T$ w.r.t. a model $M\in \TT$,
denoted $\alpha \goto{t}_M \alpha'$,
if there is some guard assignment $\beta$ s.t. $\beta \models_M t$,
$\beta(v^r) = \alpha(v)$ and
$\beta(v^w) = \alpha'(v)$ for all $v \in \mc{V}$.
\end{definition}

A \emph{$\TT$-run} of $\BB$ is a pair $(M,\rho)$ of a model $M\in \TT$
and a sequence of steps of the form
$\rho\colon \alpha_0
\goto{t_1}_M  \alpha_1
\goto{t_2}_M \dots
\goto{t_n}_M \alpha_n$ where $\alpha_0(v) = I(v)^M$ for all $v\in V$.
Note that given $M$, the initial assignment $\alpha_0$ of a run is uniquely determined by the 
initializer $I$ of $\BB$.

\subsubsection*{Encoding.}
We now show how DMTs can be modeled via suitable first-order automata. The idea is that the data variables in $\mc{V}$ are used to define a signature $\Gamma$ containing only constants, and that each constraint in $Tr$ should be thought as an \emph{action} that triggers one of the possible evolutions of the system. The structure of $\Gamma$ implies, in particular, that the encoding of a DMT will be a \emph{data-control} first-order automata. 

More specifically, the \emph{encoding} of a DMT $\BB = \dmttuple$ into a first-order automata $\autom=\seq{\Sigma,\Gamma,\phi_0,\phi_T, \phi_F}$ works as follows:
\begin{itemize}
    \item we define $\Gamma:=\{C_{\mc{V}}\}$, where $C_{\mc{V}}$ is a set of non-rigid constants $c_{v}$ corresponding to each data-variable $v$ in $\mc{V}$. Notice that such a $\Gamma$ contains only constants.
    \item $\Sigma:= \Pi \cup \mc{P}_{act}$, where $\mc{P}_{act}$ is a set of
    propositions $p_a$ corresponding to each constraint $a$ (each 'action') in
    $Tr(\mc{V}^r,\mc{V}^w)$.
    \item $\phi_0:=\bigwedge_{v\in \mc{V}} c_v=c_i$, where $c_i$ are the rigid constants in $\Pi$ assigned by $\mc{I}$ to each $v \in \mc{V}$. Notice that $\phi_0$ is a first-order $\Gamma$-sentence.
    \item $\phi_T:= \bigvee_{p_a\in \mc{P}_{act}} (p_a\to a[v^r/v,v^w/v'])$, where $a[v^r/v,v^w/v']$ is the first-order $\Gamma\cup\Sigma \cup \Gamma'$-sentence obtained by substituting every occurrence of read variable $v^r$ with $v$ and every occurrence of write variable $v^w$ with its stepped version $v'$.
    \item there is no notion of \emph{acceptance condition} in DMTs, so $\phi_F$ can be set arbitrarily.
\end{itemize}

One can prove the following.
\begin{theorem}
Let $\BB = \dmttuple$ be a DMT, and let  $\autom=\seq{\Sigma,\Gamma,\phi_0,\phi_T, \phi_F}$ be the data-control first-order automaton \emph{encoding} $\BB$ as above. Then, a \emph{$\TT$-run} of $\BB$ is a \emph{$\TT$-run} of the encoding $\autom$. 
\end{theorem}
\begin{proof}
    The proof follows immediately from the definition of the encoding.    
\end{proof}